\documentclass[conference]{IEEEtran-v17}
\pdfoutput=1 

\usepackage{xspace}
\usepackage[nosort]{cite}        
\usepackage{url} 
\usepackage[cmex10]{amsmath}
\usepackage{bbm}
\usepackage{graphicx}
\usepackage{paralist}
\usepackage{fancyref}
\usepackage[stretch=16,shrink=16,step=4]{microtype}
\newlength{\figwidth}
\makeatletter
\def\@IEEEinterspaceratioM{0.265}
\def\@IEEEinterspaceMINratioM{0.1651}
\def\@IEEEinterspaceMAXratioM{0.38}

\def\@IEEEinterspaceratioB{0.31}
\def\@IEEEinterspaceMINratioB{0.19}
\def\@IEEEinterspaceMAXratioB{0.38}
\@IEEEtunefonts
\makeatother

\usepackage{amsmath}
\usepackage{amssymb}
\usepackage{amsfonts}
\usepackage{mathrsfs}
\usepackage{xspace}
\usepackage{bm}
\usepackage{fancyref}
\usepackage{textcomp}
\usepackage[Symbol]{upgreek}

\newcommand{\safemath}[2]{\newcommand{#1}{\ensuremath{#2}\xspace}}

\newcommand{\ssa}{\mathsf{a}}
\newcommand{\ssb}{\mathsf{b}}
\newcommand{\ssc}{\mathsf{c}}
\newcommand{\ssd}{\mathsf{d}}
\newcommand{\sse}{\mathsf{e}}
\newcommand{\ssf}{\mathsf{f}}
\newcommand{\ssg}{\mathsf{g}}
\newcommand{\ssh}{\mathsf{h}}
\newcommand{\ssi}{\mathsf{i}}
\newcommand{\ssj}{\mathsf{j}}
\newcommand{\ssk}{\mathsf{k}}
\newcommand{\ssl}{\mathsf{l}}
\newcommand{\ssm}{\mathsf{m}}
\newcommand{\ssn}{\mathsf{n}}
\newcommand{\sso}{\mathsf{o}}
\newcommand{\ssp}{\mathsf{p}}
\newcommand{\ssq}{\mathsf{q}}
\newcommand{\ssr}{\mathsf{r}}
\newcommand{\sss}{\mathsf{s}}
\newcommand{\sst}{\mathsf{t}}
\newcommand{\ssu}{\mathsf{u}}
\newcommand{\ssv}{\mathsf{v}}
\newcommand{\ssw}{\mathsf{w}}
\newcommand{\ssx}{\mathsf{x}}
\newcommand{\ssy}{\mathsf{y}}
\newcommand{\ssz}{\mathsf{z}}

\safemath{\bmsa}{\bm{\ssa}}
\safemath{\bmsb}{\bm{\ssb}}
\safemath{\bmsc}{\bm{\ssc}}
\safemath{\bmsd}{\bm{\ssd}}
\safemath{\bmse}{\bm{\sse}}
\safemath{\bmsf}{\bm{\ssf}}
\safemath{\bmsg}{\bm{\ssg}}
\safemath{\bmsh}{\bm{\ssh}}
\safemath{\bmsi}{\bm{\ssi}}
\safemath{\bmsj}{\bm{\ssj}}
\safemath{\bmsk}{\bm{\ssk}}
\safemath{\bmsl}{\bm{\ssl}}
\safemath{\bmsm}{\bm{\ssm}}
\safemath{\bmsn}{\bm{\ssn}}
\safemath{\bmso}{\bm{\sso}}
\safemath{\bmsp}{\bm{\ssp}}
\safemath{\bmsq}{\bm{\ssq}}
\safemath{\bmsr}{\bm{\ssr}}
\safemath{\bmss}{\bm{\sss}}
\safemath{\bmst}{\bm{\sst}}
\safemath{\bmsu}{\bm{\ssu}}
\safemath{\bmsv}{\bm{\ssv}}
\safemath{\bmsw}{\bm{\ssw}}
\safemath{\bmsx}{\bm{\ssx}}
\safemath{\bmsy}{\bm{\ssy}}
\safemath{\bmsz}{\bm{\ssz}}

\bmdefine{\bmualphad}{\upalpha}
\bmdefine{\bmubetad}{\upbeta}
\bmdefine{\bmuchid}{\upchi}
\bmdefine{\bmudeltad}{\updelta}
\bmdefine{\bmuepsilond}{\upepsilon}
\bmdefine{\bmuvarepsilond}{\upvarepsilon}
\bmdefine{\bmuetad}{\upeta}
\bmdefine{\bmugammad}{\upgamma}
\bmdefine{\bmuiotad}{\upiota}
\bmdefine{\bmukappad}{\upkappa}
\bmdefine{\bmulambdad}{\uplambda}
\bmdefine{\bmumud}{\upmu}
\bmdefine{\bmunud}{\upnu}
\bmdefine{\bmuomegad}{\upomega}
\bmdefine{\bmuphid}{\upphi}
\bmdefine{\bmuvarphid}{\upvarphi}
\bmdefine{\bmupid}{\uppi}
\bmdefine{\bmuvarpid}{\upvarpi}
\bmdefine{\bmupsid}{\uppsi}
\bmdefine{\bmurhod}{\uprho}
\bmdefine{\bmuvarrhod}{\upvarrho}
\bmdefine{\bmusigmad}{\upsigma}
\bmdefine{\bmuvarsigmad}{\upvarsigma}
\bmdefine{\bmutaud}{\uptau}
\bmdefine{\bmuthetad}{\uptheta}
\bmdefine{\bmuvarthetad}{\upvartheta}
\bmdefine{\bmuupsilond}{\upupsilon}
\bmdefine{\bmuxid}{\upxi}
\bmdefine{\bmuzetad}{\upzeta}

\safemath{\bmua}{\mathbf{a}}
\safemath{\bmub}{\mathbf{b}}
\safemath{\bmuc}{\mathbf{c}}
\safemath{\bmud}{\mathbf{d}}
\safemath{\bmue}{\mathbf{e}}
\safemath{\bmuf}{\mathbf{f}}
\safemath{\bmug}{\mathbf{g}}
\safemath{\bmuh}{\mathbf{h}}
\safemath{\bmui}{\mathbf{i}}
\safemath{\bmuj}{\mathbf{j}}
\safemath{\bmuk}{\mathbf{k}}
\safemath{\bmul}{\mathbf{l}}
\safemath{\bmum}{\mathbf{m}}
\safemath{\bmun}{\mathbf{n}}
\safemath{\bmuo}{\mathbf{o}}
\safemath{\bmup}{\mathbf{p}}
\safemath{\bmuq}{\mathbf{q}}
\safemath{\bmur}{\mathbf{r}}
\safemath{\bmus}{\mathbf{s}}
\safemath{\bmut}{\mathbf{t}}
\safemath{\bmuu}{\mathbf{u}}
\safemath{\bmuv}{\mathbf{v}}
\safemath{\bmuw}{\mathbf{w}}
\safemath{\bmux}{\mathbf{x}}
\safemath{\bmuy}{\mathbf{y}}
\safemath{\bmuz}{\mathbf{z}}

\safemath{\bmualpha}{\bmualphad}
\safemath{\bmubeta}{\bmubetad}
\safemath{\bmuchi}{\bumchid}
\safemath{\bmudelta}{\bmudeltad}
\safemath{\bmuepsilon}{\bmuepsilond}
\safemath{\bmuvarepsilon}{\bmuvarepsilond}
\safemath{\bmueta}{\bmuetad}
\safemath{\bmugamma}{\bmugammad}
\safemath{\bmuiota}{\bmuiotad}
\safemath{\bmukappa}{\bmukappad}
\safemath{\bmulambda}{\bmulambdad}
\safemath{\bmumu}{\bmumud}
\safemath{\bmunu}{\bmunud}
\safemath{\bmuomega}{\bmuomegad}
\safemath{\bmuphi}{\bmuphid}
\safemath{\bmuvarphi}{\bmuvarphid}
\safemath{\bmupi}{\bmupid}
\safemath{\bmuvarpi}{\bmuvarpid}
\safemath{\bmupsi}{\bmupsid}
\safemath{\bmurho}{\bmurhod}
\safemath{\bmuvarrho}{\bmuvarrhod}
\safemath{\bmusigma}{\bmusigmad}
\safemath{\bmuvarsigma}{\bmuvarsigmad}
\safemath{\bmutau}{\bmutaud}
\safemath{\bmutheta}{\bmuthetad}
\safemath{\bmuvartheta}{\bmuvarthetad}
\safemath{\bmuupsilon}{\bmuupsilond}
\safemath{\bmuxi}{\bmuxid}
\safemath{\bmuzeta}{\bmuzetad}

\bmdefine{\bmiad}{a}
\bmdefine{\bmibd}{b}
\bmdefine{\bmicd}{c}
\bmdefine{\bmidd}{d}
\bmdefine{\bmied}{e}
\bmdefine{\bmifd}{f}
\bmdefine{\bmigd}{g}
\bmdefine{\bmihd}{h}
\bmdefine{\bmiid}{i}
\bmdefine{\bmijd}{j}
\bmdefine{\bmikd}{k}
\bmdefine{\bmild}{l}
\bmdefine{\bmimd}{m}
\bmdefine{\bmind}{n}
\bmdefine{\bmiod}{o}
\bmdefine{\bmipd}{p}
\bmdefine{\bmiqd}{q}
\bmdefine{\bmird}{r}
\bmdefine{\bmisd}{s}
\bmdefine{\bmitd}{t}
\bmdefine{\bmiud}{u}
\bmdefine{\bmivd}{v}
\bmdefine{\bmiwd}{w}
\bmdefine{\bmixd}{x}
\bmdefine{\bmiyd}{y}
\bmdefine{\bmizd}{z}

\bmdefine{\bmialphad}{\alpha}
\bmdefine{\bmibetad}{\beta}
\bmdefine{\bmichid}{\chi}
\bmdefine{\bmideltad}{\delta}
\bmdefine{\bmiepsilond}{\epsilon}
\bmdefine{\bmivarepsilond}{\varepsilon}
\bmdefine{\bmietad}{\eta}
\bmdefine{\bmigammad}{\gamma}
\bmdefine{\bmiiotad}{\iota}
\bmdefine{\bmikappad}{\kappa}
\bmdefine{\bmivarkappad}{\varkappa}
\bmdefine{\bmilambdad}{\lambda}
\bmdefine{\bmimud}{\mu}
\bmdefine{\bminud}{\nu}
\bmdefine{\bmiomegad}{\omega}
\bmdefine{\bmiphid}{\phi}
\bmdefine{\bmivarphid}{\varphi}
\bmdefine{\bmipid}{\pi}
\bmdefine{\bmivarpid}{\varpi}
\bmdefine{\bmipsid}{\psi}
\bmdefine{\bmirhod}{\rho}
\bmdefine{\bmivarrhod}{\varrho}
\bmdefine{\bmisigmad}{\sigma}
\bmdefine{\bmivarsigmad}{\varsigma}
\bmdefine{\bmitaud}{\tau}
\bmdefine{\bmithetad}{\theta}
\bmdefine{\bmivarthetad}{\vartheta}
\bmdefine{\bmiupsilond}{\upsilon}
\bmdefine{\bmixid}{\xi}
\bmdefine{\bmizetad}{\zeta}

\safemath{\bmia}{\bmiad}
\safemath{\bmib}{\bmibd}
\safemath{\bmic}{\bmicd}
\safemath{\bmid}{\bmidd}
\safemath{\bmie}{\bmied}
\safemath{\bmif}{\bmifd}
\safemath{\bmig}{\bmigd}
\safemath{\bmih}{\bmihd}
\safemath{\bmii}{\bmiid}
\safemath{\bmij}{\bmijd}
\safemath{\bmik}{\bmikd}
\safemath{\bmil}{\bmild}
\safemath{\bmim}{\bmimd}
\safemath{\bmin}{\bmind}
\safemath{\bmio}{\bmiod}
\safemath{\bmip}{\bmipd}
\safemath{\bmiq}{\bmiqd}
\safemath{\bmir}{\bmird}
\safemath{\bmis}{\bmisd}
\safemath{\bmit}{\bmitd}
\safemath{\bmiu}{\bmiud}
\safemath{\bmiv}{\bmivd}
\safemath{\bmiw}{\bmiwd}
\safemath{\bmix}{\bmixd}
\safemath{\bmiy}{\bmiyd}
\safemath{\bmiz}{\bmizd}

\safemath{\bmialpha}{\bmialphad}
\safemath{\bmibeta}{\bmibetad}
\safemath{\bmichi}{\bmichid}
\safemath{\bmidelta}{\bmideltad}
\safemath{\bmiepsilon}{\bmiepsilond}
\safemath{\bmivarepsilon}{\bmivarepsilond}
\safemath{\bmieta}{\bmietad}
\safemath{\bmigamma}{\bmigammad}
\safemath{\bmiiota}{\bmiiotad}
\safemath{\bmikappa}{\bmikappad}
\safemath{\bmivarkappa}{\bmivarkappad}
\safemath{\bmilambda}{\bmilambdad}
\safemath{\bmimu}{\bmimud}
\safemath{\bminu}{\bminud}
\safemath{\bmiomega}{\bmiomegad}
\safemath{\bmiphi}{\bmiphid}
\safemath{\bmivarphi}{\bmivarphid}
\safemath{\bmipi}{\bmipid}
\safemath{\bmivarpi}{\bmivarpid}
\safemath{\bmipsi}{\bmipsid}
\safemath{\bmirho}{\bmirhod}
\safemath{\bmivarrho}{\bmivarrhod}
\safemath{\bmisigma}{\bmisigmad}
\safemath{\bmivarsigma}{\bmivarsigmad}
\safemath{\bmitau}{\bmitaud}
\safemath{\bmitheta}{\bmithetad}
\safemath{\bmivartheta}{\bmivarthetad}
\safemath{\bmiupsilon}{\bmiupsilond}
\safemath{\bmixi}{\bmixid}
\safemath{\bmizeta}{\bmizetad}

\bmdefine{\bmuDeltad}{\Updelta}
\bmdefine{\bmuGammad}{\Upgamma}
\bmdefine{\bmuLambdad}{\Uplambda}
\bmdefine{\bmuOmegad}{\Upomega}
\bmdefine{\bmuPhid}{\Upphi}
\bmdefine{\bmuPid}{\Uppi}
\bmdefine{\bmuPsid}{\Uppsi}
\bmdefine{\bmuSigmad}{\Upsigma}
\bmdefine{\bmuThetad}{\Uptheta}
\bmdefine{\bmuUpsilond}{\Upupsilon}
\bmdefine{\bmuXid}{\Upxi}

\safemath{\bmuA}{\mathbf{A}}
\safemath{\bmuB}{\mathbf{B}}
\safemath{\bmuC}{\mathbf{C}}
\safemath{\bmuD}{\mathbf{D}}
\safemath{\bmuE}{\mathbf{E}}
\safemath{\bmuF}{\mathbf{F}}
\safemath{\bmuG}{\mathbf{G}}
\safemath{\bmuH}{\mathbf{H}}
\safemath{\bmuI}{\mathbf{I}}
\safemath{\bmuJ}{\mathbf{J}}
\safemath{\bmuK}{\mathbf{K}}
\safemath{\bmuL}{\mathbf{L}}
\safemath{\bmuM}{\mathbf{M}}
\safemath{\bmuN}{\mathbf{N}}
\safemath{\bmuO}{\mathbf{O}}
\safemath{\bmuP}{\mathbf{P}}
\safemath{\bmuQ}{\mathbf{Q}}
\safemath{\bmuR}{\mathbf{R}}
\safemath{\bmuS}{\mathbf{S}}
\safemath{\bmuT}{\mathbf{T}}
\safemath{\bmuU}{\mathbf{U}}
\safemath{\bmuV}{\mathbf{V}}
\safemath{\bmuW}{\mathbf{W}}
\safemath{\bmuX}{\mathbf{X}}
\safemath{\bmuY}{\mathbf{Y}}
\safemath{\bmuZ}{\mathbf{Z}}

\safemath{\bmuZero}{\mathbf{0}}
\safemath{\bmuOne}{\mathbf{1}}

\safemath{\bmuDelta}{\bmuDeltad}
\safemath{\bmuGamma}{\bmuGammad}
\safemath{\bmuLambda}{\bmuLambdad}
\safemath{\bmuOmega}{\bmuOmegad}
\safemath{\bmuPhi}{\bmuPhid}
\safemath{\bmuPi}{\bmuPid}
\safemath{\bmuPsi}{\bmuPsid}
\safemath{\bmuSigma}{\bmuSigmad}
\safemath{\bmuTheta}{\bmuThetad}
\safemath{\bmuUpsilon}{\bmuUpsilond}
\safemath{\bmuXi}{\bmuXid}

\bmdefine{\bmiAd}{A}
\bmdefine{\bmiBd}{B}
\bmdefine{\bmiCd}{C}
\bmdefine{\bmiDd}{D}
\bmdefine{\bmiEd}{E}
\bmdefine{\bmiFd}{F}
\bmdefine{\bmiGd}{G}
\bmdefine{\bmiHd}{H}
\bmdefine{\bmiId}{I}
\bmdefine{\bmiJd}{J}
\bmdefine{\bmiKd}{K}
\bmdefine{\bmiLd}{L}
\bmdefine{\bmiMd}{M}
\bmdefine{\bmiOd}{N}
\bmdefine{\bmiPd}{O}
\bmdefine{\bmiQd}{P}
\bmdefine{\bmiRd}{R}
\bmdefine{\bmiSd}{S}
\bmdefine{\bmiTd}{T}
\bmdefine{\bmiUd}{U}
\bmdefine{\bmiVd}{V}
\bmdefine{\bmiWd}{W}
\bmdefine{\bmiXd}{X}
\bmdefine{\bmiYd}{Y}
\bmdefine{\bmiZd}{Z}

\bmdefine{\bmiDeltad}{\Delta}
\bmdefine{\bmiGammad}{\Gamma}
\bmdefine{\bmiLambdad}{\Lambda}
\bmdefine{\bmiOmegad}{\Omega}
\bmdefine{\bmiPhid}{\Phi}
\bmdefine{\bmiPid}{\Pi}
\bmdefine{\bmiPsid}{\Psi}
\bmdefine{\bmiSigmad}{\Sigma}
\bmdefine{\bmiThetad}{\Theta}
\bmdefine{\bmiUpsilond}{\Upsilon}
\bmdefine{\bmiXid}{\Xi}

\safemath{\bmiA}{\bmiAd}
\safemath{\bmiB}{\bmiBd}
\safemath{\bmiC}{\bmiCd}
\safemath{\bmiD}{\bmiDd}
\safemath{\bmiE}{\bmiEd}
\safemath{\bmiF}{\bmiFd}
\safemath{\bmiG}{\bmiGd}
\safemath{\bmiH}{\bmiHd}
\safemath{\bmiI}{\bmiId}
\safemath{\bmiJ}{\bmiJd}
\safemath{\bmiK}{\bmiKd}
\safemath{\bmiL}{\bmiLd}
\safemath{\bmiM}{\bmiMd}
\safemath{\bmiN}{\bmiNd}
\safemath{\bmiO}{\bmiOd}
\safemath{\bmiP}{\bmiPd}
\safemath{\bmiQ}{\bmiQd}
\safemath{\bmiR}{\bmiRd}
\safemath{\bmiS}{\bmiSd}
\safemath{\bmiT}{\bmiTd}
\safemath{\bmiU}{\bmiUd}
\safemath{\bmiV}{\bmiVd}
\safemath{\bmiW}{\bmiWd}
\safemath{\bmiX}{\bmiXd}
\safemath{\bmiY}{\bmiYd}
\safemath{\bmiZ}{\bmiZd}

\safemath{\bmiDelta}{\bmiDeltad}
\safemath{\bmiGamma}{\bmiGammad}
\safemath{\bmiLambda}{\bmiLambdad}
\safemath{\bmiOmega}{\bmiOmegad}
\safemath{\bmiPhi}{\bmiPhid}
\safemath{\bmiPi}{\bmiPid}
\safemath{\bmiPsi}{\bmiPsid}
\safemath{\bmiSigma}{\bmiSigmad}
\safemath{\bmiTheta}{\bmiThetad}
\safemath{\bmiUpsilon}{\bmiUpsilond}
\safemath{\bmiXi}{\bmiXid}

\safemath{\setA}{\mathcal{A}}
\safemath{\setB}{\mathcal{B}}
\safemath{\setC}{\mathcal{C}}
\safemath{\setD}{\mathcal{D}}
\safemath{\setE}{\mathcal{E}}
\safemath{\setF}{\mathcal{F}}
\safemath{\setG}{\mathcal{G}}
\safemath{\setH}{\mathcal{H}}
\safemath{\setI}{\mathcal{I}}
\safemath{\setJ}{\mathcal{J}}
\safemath{\setK}{\mathcal{K}}
\safemath{\setL}{\mathcal{L}}
\safemath{\setM}{\mathcal{M}}
\safemath{\setN}{\mathcal{N}}
\safemath{\setO}{\mathcal{O}}
\safemath{\setP}{\mathcal{P}}
\safemath{\setQ}{\mathcal{Q}}
\safemath{\setR}{\mathcal{R}}
\safemath{\setS}{\mathcal{S}}
\safemath{\setT}{\mathcal{T}}
\safemath{\setU}{\mathcal{U}}
\safemath{\setV}{\mathcal{V}}
\safemath{\setW}{\mathcal{W}}
\safemath{\setX}{\mathcal{X}}
\safemath{\setY}{\mathcal{Y}}
\safemath{\setZ}{\mathcal{Z}}
\safemath{\emptySet}{\varnothing}

\safemath{\colA}{\mathscr{A}}
\safemath{\colB}{\mathscr{B}}
\safemath{\colC}{\mathscr{C}}
\safemath{\colD}{\mathscr{D}}
\safemath{\colE}{\mathscr{E}}
\safemath{\colF}{\mathscr{F}}
\safemath{\colG}{\mathscr{G}}
\safemath{\colH}{\mathscr{H}}
\safemath{\colI}{\mathscr{I}}
\safemath{\colJ}{\mathscr{J}}
\safemath{\colK}{\mathscr{K}}
\safemath{\colL}{\mathscr{L}}
\safemath{\colM}{\mathscr{M}}
\safemath{\colN}{\mathscr{N}}
\safemath{\colO}{\mathscr{O}}
\safemath{\colP}{\mathscr{P}}
\safemath{\colQ}{\mathscr{Q}}
\safemath{\colR}{\mathscr{R}}
\safemath{\colS}{\mathscr{S}}
\safemath{\colT}{\mathscr{T}}
\safemath{\colU}{\mathscr{U}}
\safemath{\colV}{\mathscr{V}}
\safemath{\colW}{\mathscr{W}}
\safemath{\colX}{\mathscr{X}}
\safemath{\colY}{\mathscr{Y}}
\safemath{\colZ}{\mathscr{Z}}

\safemath{\opA}{\mathbb{A}}
\safemath{\opB}{\mathbb{B}}
\safemath{\opC}{\mathbb{C}}
\safemath{\opD}{\mathbb{D}}
\safemath{\opE}{\mathbb{E}}
\safemath{\opF}{\mathbb{F}}
\safemath{\opG}{\mathbb{G}}
\safemath{\opH}{\mathbb{H}}
\safemath{\opI}{\mathbb{I}}
\safemath{\opJ}{\mathbb{J}}
\safemath{\opK}{\mathbb{K}}
\safemath{\opL}{\mathbb{L}}
\safemath{\opM}{\mathbb{M}}
\safemath{\opN}{\mathbb{N}}
\safemath{\opO}{\mathbb{O}}
\safemath{\opP}{\mathbb{P}}
\safemath{\opQ}{\mathbb{Q}}
\safemath{\opR}{\mathbb{R}}
\safemath{\opS}{\mathbb{S}}
\safemath{\opT}{\mathbb{T}}
\safemath{\opU}{\mathbb{U}}
\safemath{\opV}{\mathbb{V}}
\safemath{\opW}{\mathbb{W}}
\safemath{\opX}{\mathbb{X}}
\safemath{\opY}{\mathbb{Y}}
\safemath{\opZ}{\mathbb{Z}}
\safemath{\opZero}{\mathbb{O}}
\safemath{\identityop}{\opI}

\safemath{\sca}{a}
\safemath{\scb}{b}
\safemath{\scc}{c}
\safemath{\scd}{d}
\safemath{\sce}{e}
\safemath{\scf}{f}
\safemath{\scg}{g}
\safemath{\sch}{h}
\safemath{\sci}{i}
\safemath{\scj}{j}
\safemath{\sck}{k}
\safemath{\scl}{l}
\safemath{\scm}{m}
\safemath{\scn}{n}
\safemath{\sco}{o}
\safemath{\scp}{p}
\safemath{\scq}{q}
\safemath{\scr}{r}
\safemath{\scs}{s}
\safemath{\sct}{t}
\safemath{\scu}{u}
\safemath{\scv}{v}
\safemath{\scw}{w}
\safemath{\scx}{x}
\safemath{\scy}{y}
\safemath{\scz}{z}

\safemath{\scA}{A}
\safemath{\scB}{B}
\safemath{\scC}{C}
\safemath{\scD}{D}
\safemath{\scE}{E}
\safemath{\scF}{F}
\safemath{\scG}{G}
\safemath{\scH}{H}
\safemath{\scI}{I}
\safemath{\scJ}{J}
\safemath{\scK}{K}
\safemath{\scL}{L}
\safemath{\scM}{M}
\safemath{\scN}{N}
\safemath{\scO}{O}
\safemath{\scP}{P}
\safemath{\scQ}{Q}
\safemath{\scR}{R}
\safemath{\scS}{S}
\safemath{\scT}{T}
\safemath{\scU}{U}
\safemath{\scV}{V}
\safemath{\scW}{W}
\safemath{\scX}{X}
\safemath{\scY}{Y}
\safemath{\scZ}{Z}

\safemath{\scalpha}{\alpha}
\safemath{\scbeta}{\beta}
\safemath{\scchi}{\chi}
\safemath{\scdelta}{\delta}
\safemath{\scepsilon}{\epsilon}
\safemath{\scvarepsilon}{\varepsilon}
\safemath{\sceta}{\eta}
\safemath{\scgamma}{\gamma}
\safemath{\sciota}{\iota}
\safemath{\sckappa}{\kappa}
\safemath{\scvarkappa}{\varkappa}
\safemath{\sclambda}{\lambda}
\safemath{\scmu}{\mu}
\safemath{\scnu}{\nu}
\safemath{\scomega}{\omega}
\safemath{\scphi}{\phi}
\safemath{\scvarphi}{\varphi}
\safemath{\scpi}{\pi}
\safemath{\scvarpi}{\varpi}
\safemath{\scpsi}{\psi}
\safemath{\scrho}{\rho}
\safemath{\scvarrho}{\varrho}
\safemath{\scsigma}{\sigma}
\safemath{\scvarsigma}{\varsigma}
\safemath{\sctau}{\tau}
\safemath{\sctheta}{\theta}
\safemath{\scvartheta}{\vartheta}
\safemath{\scupsilon}{\upsilon}
\safemath{\scxi}{\xi}
\safemath{\sczeta}{\zeta}

\safemath{\veca}{\mathbf{a}}
\safemath{\vecb}{\mathbf{b}}
\safemath{\vecc}{\mathbf{c}}
\safemath{\vecd}{\mathbf{d}}
\safemath{\vece}{\mathbf{e}}
\safemath{\vecf}{\mathbf{f}}
\safemath{\vecg}{\mathbf{g}}
\safemath{\vech}{\mathbf{h}}
\safemath{\veci}{\mathbf{i}}
\safemath{\vecj}{\mathbf{j}}
\safemath{\veck}{\mathbf{k}}
\safemath{\vecl}{\mathbf{l}}
\safemath{\vecm}{\mathbf{m}}
\safemath{\vecn}{\mathbf{n}}
\safemath{\veco}{\mathbf{o}}
\safemath{\vecp}{\mathbf{p}}
\safemath{\vecq}{\mathbf{q}}
\safemath{\vecr}{\mathbf{r}}
\safemath{\vecs}{\mathbf{s}}
\safemath{\vect}{\mathbf{t}}
\safemath{\vecu}{\mathbf{u}}
\safemath{\vecv}{\mathbf{v}}
\safemath{\vecw}{\mathbf{w}}
\safemath{\vecx}{\mathbf{x}}
\safemath{\vecy}{\mathbf{y}}
\safemath{\vecz}{\mathbf{z}}
\safemath{\veczero}{\mathbf{0}}
\safemath{\vecone}{\mathbf{1}}

\safemath{\vecalpha}{\upalpha}
\safemath{\vecbeta}{\upbeta}
\safemath{\vecchi}{\upchi}
\safemath{\vecdelta}{\updelta}
\safemath{\vecepsilon}{\upepsilon}
\safemath{\vecvarepsilon}{\upvarepsilon}
\safemath{\veceta}{\upeta}
\safemath{\vecgamma}{\upgamma}
\safemath{\veciota}{\upiota}
\safemath{\veckappa}{\upkappa}
\safemath{\veclambda}{\uplambda}
\safemath{\vecmu}{\text{\textmu}}
\safemath{\vecnu}{\upnu}
\safemath{\vecomega}{\upomega}
\safemath{\vecphi}{\upphi}
\safemath{\vecvarphi}{\upvarphi}
\safemath{\vecpi}{\uppi}
\safemath{\vecvarpi}{\upvarpi}
\safemath{\vecpsi}{\uppsi}
\safemath{\vecrho}{\uprho}
\safemath{\vecvarrho}{\upvarrho}
\safemath{\vecsigma}{\upsigma}
\safemath{\vecvarsigma}{\upvarsigma}
\safemath{\vectau}{\uptau}
\safemath{\vectheta}{\uptheta}
\safemath{\vecvartheta}{\upvartheta}
\safemath{\vecupsilon}{\upupsilon}
\safemath{\vecxi}{\upxi}
\safemath{\veczeta}{\upzeta}

\safemath{\vecac}{a}
\safemath{\vecbc}{b}
\safemath{\veccc}{c}
\safemath{\vecdc}{d}
\safemath{\vecec}{e}
\safemath{\vecfc}{f}
\safemath{\vecgc}{g}
\safemath{\vechc}{h}
\safemath{\vecic}{i}
\safemath{\vecjc}{j}
\safemath{\veckc}{k}
\safemath{\veclc}{l}
\safemath{\vecmc}{m}
\safemath{\vecnc}{n}
\safemath{\vecoc}{o}
\safemath{\vecpc}{p}
\safemath{\vecqc}{q}
\safemath{\vecrc}{r}
\safemath{\vecsc}{s}
\safemath{\vectc}{t}
\safemath{\vecuc}{u}
\safemath{\vecvc}{v}
\safemath{\vecwc}{w}
\safemath{\vecxc}{x}
\safemath{\vecyc}{y}
\safemath{\veczc}{z}

\safemath{\matA}{\mathbf{A}}
\safemath{\matB}{\mathbf{B}}
\safemath{\matC}{\mathbf{C}}
\safemath{\matD}{\mathbf{D}}
\safemath{\matE}{\mathbf{E}}
\safemath{\matF}{\mathbf{F}}
\safemath{\matG}{\mathbf{G}}
\safemath{\matH}{\mathbf{H}}
\safemath{\matI}{\mathbf{I}}
\safemath{\matJ}{\mathbf{J}}
\safemath{\matK}{\mathbf{K}}
\safemath{\matL}{\mathbf{L}}
\safemath{\matM}{\mathbf{M}}
\safemath{\matN}{\mathbf{N}}
\safemath{\matO}{\mathbf{O}}
\safemath{\matP}{\mathbf{P}}
\safemath{\matQ}{\mathbf{Q}}
\safemath{\matR}{\mathbf{R}}
\safemath{\matS}{\mathbf{S}}
\safemath{\matT}{\mathbf{T}}
\safemath{\matU}{\mathbf{U}}
\safemath{\matV}{\mathbf{V}}
\safemath{\matW}{\mathbf{W}}
\safemath{\matX}{\mathbf{X}}
\safemath{\matY}{\mathbf{Y}}
\safemath{\matZ}{\mathbf{Z}}
\safemath{\matzero}{\mathbf{0}}

\safemath{\matDelta}{\Updelta}
\safemath{\matGamma}{\Upgammma}
\safemath{\matLambda}{\Uplambda}
\safemath{\matOmega}{\Upomega}
\safemath{\matPhi}{\Upphi}
\safemath{\matPi}{\Uppi}
\safemath{\matPsi}{\Uppsi}
\safemath{\matSigma}{\Upsigma}
\safemath{\matTheta}{\Uptheta}
\safemath{\matUpsilon}{\Upupsilon}
\safemath{\matXi}{\Upxi}

\safemath{\matidentity}{\matI}
\safemath{\matone}{\matO}

\safemath{\matAc}{a}
\safemath{\matBc}{b}
\safemath{\matCc}{c}
\safemath{\matDc}{d}
\safemath{\matEc}{e}
\safemath{\matFc}{f}
\safemath{\matGc}{g}
\safemath{\matHc}{h}
\safemath{\matIc}{i}
\safemath{\matJc}{j}
\safemath{\matKc}{k}
\safemath{\matLc}{l}
\safemath{\matMc}{m}
\safemath{\matNc}{n}
\safemath{\matOc}{o}
\safemath{\matPc}{p}
\safemath{\matQc}{q}
\safemath{\matRc}{r}
\safemath{\matSc}{s}
\safemath{\matTc}{t}
\safemath{\matUc}{u}
\safemath{\matVc}{v}
\safemath{\matWc}{w}
\safemath{\matXc}{x}
\safemath{\matYc}{y}
\safemath{\matZc}{z}

\safemath{\rnda}{\bmia}
\safemath{\rndb}{\bmib}
\safemath{\rndc}{\bmic}
\safemath{\rndd}{\bmid}
\safemath{\rnde}{\bmie}
\safemath{\rndf}{\bmif}
\safemath{\rndg}{\bmig}
\safemath{\rndh}{\bmih}
\safemath{\rndi}{\bmii}
\safemath{\rndj}{\bmij}
\safemath{\rndk}{\bmik}
\safemath{\rndl}{\bmil}
\safemath{\rndm}{\bmim}
\safemath{\rndn}{\bmin}
\safemath{\rndo}{\bmio}
\safemath{\rndp}{\bmip}
\safemath{\rndq}{\bmiq}
\safemath{\rndr}{\bmir}
\safemath{\rnds}{\bmis}
\safemath{\rndt}{\bmit}
\safemath{\rndu}{\bmiu}
\safemath{\rndv}{\bmiv}
\safemath{\rndw}{\bmiw}
\safemath{\rndx}{\bmix}
\safemath{\rndy}{\bmiy}
\safemath{\rndz}{\bmiz}

\safemath{\rndA}{\bmiA}
\safemath{\rndB}{\bmiB}
\safemath{\rndC}{\bmiC}
\safemath{\rndD}{\bmiD}
\safemath{\rndE}{\bmiE}
\safemath{\rndF}{\bmiF}
\safemath{\rndG}{\bmiG}
\safemath{\rndH}{\bmiH}
\safemath{\rndI}{\bmiI}
\safemath{\rndJ}{\bmiJ}
\safemath{\rndK}{\bmiK}
\safemath{\rndL}{\bmiL}
\safemath{\rndM}{\bmiM}
\safemath{\rndN}{\bmiN}
\safemath{\rndO}{\bmiO}
\safemath{\rndP}{\bmiP}
\safemath{\rndQ}{\bmiQ}
\safemath{\rndR}{\bmiR}
\safemath{\rndS}{\bmiS}
\safemath{\rndT}{\bmiT}
\safemath{\rndU}{\bmiU}
\safemath{\rndV}{\bmiV}
\safemath{\rndW}{\bmiW}
\safemath{\rndX}{\bmiX}
\safemath{\rndY}{\bmiY}
\safemath{\rndZ}{\bmiZ}

\safemath{\rndalpha}{\bmialpha}
\safemath{\rndbeta}{\bmibeta}
\safemath{\rndchi}{\bmichi}
\safemath{\rnddelta}{\bmidelta}
\safemath{\rndepsilon}{\bmiepsilon}
\safemath{\rndvarepsilon}{\bmivarepsilon}
\safemath{\rndeta}{\bmieta}
\safemath{\rndgamma}{\bmigamma}
\safemath{\rndiota}{\bmiiota}
\safemath{\rndkappa}{\bmikappa}
\safemath{\rndlambda}{\bmilambda}
\safemath{\rndmu}{\bmimu}
\safemath{\rndnu}{\bminu}
\safemath{\rndomega}{\bmiomega}
\safemath{\rndphi}{\bmiphi}
\safemath{\rndvarphi}{\bmivarphi}
\safemath{\rndpi}{\bmipi}
\safemath{\rndvarpi}{\bmivarpi}
\safemath{\rndpsi}{\bmipsi}
\safemath{\rndrho}{\bmirho}
\safemath{\rndvarrho}{\bmivarrho}
\safemath{\rndsigma}{\bmisigma}
\safemath{\rndvarsigma}{\bmivarsigma}
\safemath{\rndtau}{\bmitau}
\safemath{\rndtheta}{\bmitheta}
\safemath{\rndvartheta}{\bmivartheta}
\safemath{\rndupsilon}{\bmiupsilon}
\safemath{\rndxi}{\bmixi}
\safemath{\rndzeta}{\bmizeta}

\safemath{\rveca}{\bmua}
\safemath{\rvecb}{\bmub}
\safemath{\rvecc}{\bmuc}
\safemath{\rvecd}{\bmud}
\safemath{\rvece}{\bmue}
\safemath{\rvecf}{\bmuf}
\safemath{\rvecg}{\bmug}
\safemath{\rvech}{\bmuh}
\safemath{\rveci}{\bmui}
\safemath{\rvecj}{\bmuj}
\safemath{\rveck}{\bmuk}
\safemath{\rvecl}{\bmul}
\safemath{\rvecm}{\bmum}
\safemath{\rvecn}{\bmun}
\safemath{\rveco}{\bmuo}
\safemath{\rvecp}{\bmup}
\safemath{\rvecq}{\bmuq}
\safemath{\rvecr}{\bmur}
\safemath{\rvecs}{\bmus}
\safemath{\rvect}{\bmut}
\safemath{\rvecu}{\bmuu}
\safemath{\rvecv}{\bmuv}
\safemath{\rvecw}{\bmuw}
\safemath{\rvecx}{\bmux}
\safemath{\rvecy}{\bmuy}
\safemath{\rvecz}{\bmuz}

\safemath{\rvecalpha}{\bmualpha}
\safemath{\rvecbeta}{\bmubeta}
\safemath{\rvecchi}{\bmuchi}
\safemath{\rvecdelta}{\bmudelta}
\safemath{\rvecepsilon}{\bmuepsilon}
\safemath{\rvecvarepsilon}{\bmuvarepsilon}
\safemath{\rveceta}{\bmueta}
\safemath{\rvecgamma}{\bmugamma}
\safemath{\rveciota}{\bmuiota}
\safemath{\rveckappa}{\bmukappa}
\safemath{\rveclambda}{\bmulambda}
\safemath{\rvecmu}{\bmumu}
\safemath{\rvecnu}{\bmunu}
\safemath{\rvecomega}{\bmuomega}
\safemath{\rvecphi}{\bmuphi}
\safemath{\rvecvarphi}{\bmuvarphi}
\safemath{\rvecpi}{\bmupi}
\safemath{\rvecvarpi}{\bmuvarpi}
\safemath{\rvecpsi}{\bmupsi}
\safemath{\rvecrho}{\bmurho}
\safemath{\rvecvarrho}{\bmuvarrho}
\safemath{\rvecsigma}{\bmusigma}
\safemath{\rvecvarsigma}{\bmuvarsigma}
\safemath{\rvectau}{\bmutau}
\safemath{\rvectheta}{\bmutheta}
\safemath{\rvecvartheta}{\bmuvartheta}
\safemath{\rvecupsilon}{\bmuupsilon}
\safemath{\rvecxi}{\bmuxi}
\safemath{\rveczeta}{\bmuzeta}

\safemath{\rmatA}{\bmuA}
\safemath{\rmatB}{\bmuB}
\safemath{\rmatC}{\bmuC}
\safemath{\rmatD}{\bmuD}
\safemath{\rmatE}{\bmuE}
\safemath{\rmatF}{\bmuF}
\safemath{\rmatG}{\bmuG}
\safemath{\rmatH}{\bmuH}
\safemath{\rmatI}{\bmuI}
\safemath{\rmatJ}{\bmuJ}
\safemath{\rmatK}{\bmuK}
\safemath{\rmatL}{\bmuL}
\safemath{\rmatM}{\bmuM}
\safemath{\rmatN}{\bmuN}
\safemath{\rmatO}{\bmuO}
\safemath{\rmatP}{\bmuP}
\safemath{\rmatQ}{\bmuQ}
\safemath{\rmatR}{\bmuR}
\safemath{\rmatS}{\bmuS}
\safemath{\rmatT}{\bmuT}
\safemath{\rmatU}{\bmuU}
\safemath{\rmatV}{\bmuV}
\safemath{\rmatW}{\bmuW}
\safemath{\rmatX}{\bmuX}
\safemath{\rmatY}{\bmuY}
\safemath{\rmatZ}{\bmuZ}

\safemath{\rmatDelta}{\bmuDelta}
\safemath{\rmatGamma}{\bmuGamma}
\safemath{\rmatLambda}{\bmuLambda}
\safemath{\rmatOmega}{\bmuOmega}
\safemath{\rmatPhi}{\bmuPhi}
\safemath{\rmatPi}{\bmuPi}
\safemath{\rmatPsi}{\bmuPsi}
\safemath{\rmatSigma}{\bmuSigma}
\safemath{\rmatTheta}{\bmuTheta}
\safemath{\rmatUpsilon}{\bmuUpsilon}
\safemath{\rmatXi}{\bmuXi}

\newenvironment{textbmatrix}{	\setlength{\arraycolsep}{2.5pt}%
								\big[\begin{matrix}}{\end{matrix}\big]%
								\raisebox{0.08ex}{\vphantom{M}}}

 \def\btm{\begin{textbmatrix}}
 \def\etm{\end{textbmatrix}}

\newcommand{\lefto}{\mathopen{}\left}

\DeclareMathOperator{\diag}{diag}			% diagonal matrix
\DeclareMathOperator{\adj}{adj}				% adjunct matrix
\DeclareMathOperator{\kron}{\otimes}			% Kroneker Product
\DeclareMathOperator{\Exop}{\opE}			% expectation operator
\DeclareMathOperator{\landauO}{\mathcal{O}}

\safemath{\fun}{\scf}						% generic scalar function
\safemath{\altfun}{\scg}
\safemath{\aaltfun}{\sch}
\safemath{\bel}{\sce}					% basis element
\safemath{\altbel}{\sce}					% alternative basis element
\safemath{\frel}{g}					% frame element
\safemath{\altfrel}{g}					% alternative frame element
\safemath{\dfrel}{\tilde{g}}					% dual frame element
\safemath{\altdfrel}{\tilde{g}}					% alternative dual frame element

\safemath{\mat}{\matA}						% generic matrix
\safemath{\matc}{\matAc}						% components of a generic matrix
\safemath{\altmat}{\matB}						% alternative generic matrix
\safemath{\altmatc}{\matBc}						% alternative generic matrix

\safemath{\vectr}{\vecu}						% generic vector
\safemath{\vectrc}{\vecuc}						% components of a generic vector
\safemath{\altvectr}{\vecv}						% alternative generic vector
\safemath{\altvectrc}{\vecvc}						% components of an alternative generic vector

\newcommand{\nullspace}{\setN}	 			% nullspace
\newcommand{\Ex}[2]{\ensuremath{\Exop_{#1}\lefto[#2\right]}} 	% expectation
\newcommand{\abs}[1]{\left\lvert#1\right\rvert}		% absolute value
\newcommand{\card}[1]{\lvert#1\rvert}			% cardinality of a set
\newcommand{\union}{\cup}					% set union

\newcommand{\vecnorm}[1]{\lVert#1\rVert}		% vector norm
\newcommand{\conj}[1]{\ensuremath{\overline{#1}}} 	% conjugate 		
\newcommand{\tp}[1]{\ensuremath{#1^{\mathsf{T}}}} 		% transpose
\newcommand{\herm}[1]{\ensuremath{#1^{\mathsf{H}}}} 	% hermitian transpose (Hilbert adjoint) //modified by Christoph Bunte
\newcommand{\inv}[1]{\ensuremath{#1^{-1}}} 	% inverse
\safemath{\dirac}{\delta}					% Dirac delta
\safemath{\diracp}{\dirac(\time)}			% 	''	parametrized
\safemath{\krond}{\dirac}					% Kronecker delta
\safemath{\upto}{\uparrow}
\safemath{\downto}{\downarrow}
\safemath{\iu}{i}							% imaginary unit
\safemath{\maj}{\succ}
\newcommand{\dftmat}[1]{\matF_{#1}}			% DFT matrix
\newcommand{\diagd}[1]{D(#1)}               % diagonalized matrix
\safemath{\mdft}{\dftmat{}}					% 	''
\safemath{\runity}{\beta}					% root of unity
\safemath{\eval}{\lambda}					% eigenvalue
\safemath{\veval}{\veclambda}				% eigenvalue vector
\safemath{\littleo}{\sco}					% Landau\s little o

\let\im\undefined
\safemath{\re}{\mathfrak{Re}}				% real part
\safemath{\im}{\mathfrak{Im}}				% imaginary part
\safemath{\euclidspace}{\complexset}			% Euclidean space
\safemath{\confunspace}{\setC}				% space of continuous functions
\newcommand{\banachseqspace}[1]{l^{#1}}		% Banach sequence space
\safemath{\hilseqspace}{\banachseqspace{2}}	% Hilbert sequence space
\newcommand{\banachfunspace}[1]{\setL^{#1}}	% Banach function space
\safemath{\hilfunspace}{\banachfunspace{2}}	% Hilbert function space
\safemath{\schwarzspace}{\setS}				% Schwarz space

\newcommand{\hadj}[1]{#1^{\star}}			% Hilbert adjoint operator

\safemath{\SNR}{\rho} 				% signal to noise ratio
\safemath{\No}{N_0}							% noise spectral density
\safemath{\Es}{E_s}							% energy per symbol
\safemath{\Eb}{E_b}							% energy per bit
\safemath{\EbNo}{\frac{\Eb}{\No}}
\safemath{\EsNo}{\frac{\Es}{\No}}

\let\time\undefined
\safemath{\time}{\sct}						% continuous time
\safemath{\dtime}{\sck}						% discrete time
\safemath{\delay}{\sctau}					% continuous delay
\safemath{\ddelay}{\scl}						% discrete delay
\safemath{\doppler}{\scnu}					% continuous doppler
\safemath{\ddoppler}{\scm}					% discrete doppler
\safemath{\freq}{\scf}						% frequency
\safemath{\dfreq}{\scn}						% discrete frequency
\safemath{\Dtime}{\Delta\time}
\safemath{\Dfreq}{\Delta\freq}
\safemath{\Ddtime}{\Delta\dtime}
\safemath{\Ddfreq}{\Delta\dfreq}
\safemath{\bandwidth}{\scB}
\safemath{\maxdoppler}{\doppler_{0}}			% maximum Doppler shift
\safemath{\maxdelay}{\delay_{0}}				% maximum delay
\safemath{\spread}{\Delta_{\CHop}}			% total channel spread
\DeclareMathOperator{\CHop}{\ensuremath{\opH}} % channel operator
\safemath{\kernel}{\rndk_{\CHop}}			% operator kernel
\safemath{\kernelp}{\kernel(\time,\time')}	% 	''	parametrized
\safemath{\tvir}{\rndh_{\CHop}}				% time-varying impulse response
\safemath{\tvirp}{\tvir(\time,\delay)}		%	''	parametrized
\safemath{\tvirc}{\conj{\rndh}_{\CHop}}		% 	''	parametrized
\safemath{\tvtf}{\rndl_{\CHop}}				% time-varying transfer function
\safemath{\tvtfp}{\tvtf(\time,\freq)}			%	''	parametrized
\safemath{\tvtfc}{\conj{\rndl}_{\CHop}}		%	''	parametrized
\safemath{\spf}{\rnds_{\CHop}}				% spreading function
\safemath{\spfp}{\spf(\doppler,\delay)}		%	''	parametrized
\safemath{\spfc}{\conj{\rnds}_{\CHop}}		%	''	parametrized
\safemath{\bff}{\rndb_{\CHop}}				% bi-freuqency function
\safemath{\bffp}{\bff(\doppler,\freq)}		%	''	parametrized

\safemath{\irc}{\scr_{\rndh}}				% impulse response correlation fn.
\safemath{\tfc}{\scr_{\rndl}}				% time-frequency correlation fn.
\safemath{\spc}{\scr_{\rnds}}				% spreading fn. correlation fn.
\safemath{\bfc}{\scr_{\rndb}}				% bi-frequency correlation fn.

\safemath{\scaf}{\scc_{\rnds}}				% scattering function
\safemath{\scafp}{\scaf(\doppler,\delay)}		% 	''	parametrized
\safemath{\ccf}{\scc_{\rndl}}				% WSSUS tvtf correlation
\safemath{\ccfp}{\ccf(\Dtime,\Dfreq)}			% 	''	parametrized
\safemath{\cic}{\scc_{\rndh}}				% WSSUS tvir correlation
\safemath{\cicp}{\cic(\Dtime,\delay)}			% 	''	parametrized

\safemath{\mi}{\scI}							% mutual information
\safemath{\capacity}{\scC}					% capacity

\newcommand{\iid}{i.i.d.\@\xspace}

\DeclareMathOperator{\Prob}{\opP}		% probability of an event
\newcommand{\pdf}[1]{\scf_{#1}}			% probability density function
\newcommand{\diffent}{\sch}				% differential entropy
\safemath{\normal}{\mathcal{N}}			% normal distribution
\safemath{\jpg}{\mathcal{CN}}			% jointly proper Gaussian
\safemath{\mchain}{\leftrightarrow}		% Markov chain
\newcommand{\given}{\,\vert\,}				% conditioning
\safemath{\dB}{\,\mathrm{dB}}
\safemath{\dBm}{\,\mathrm{dBm}}
\safemath{\Hz}{\,\mathrm{Hz}}
\safemath{\kHz}{\,\mathrm{kHz}}
\safemath{\MHz}{\,\mathrm{MHz}}
\safemath{\GHz}{\,\mathrm{GHz}}
\safemath{\s}{\,\mathrm{s}}
\safemath{\ms}{\,\mathrm{ms}}
\safemath{\mus}{\,\mathrm{\text{\textmu}s}}
\safemath{\ns}{\,\mathrm{ns}}
\safemath{\ps}{\,\mathrm{ps}}
\safemath{\meter}{\,\mathrm{m}}
\safemath{\mm}{\,\mathrm{mm}}
\safemath{\cm}{\,\mathrm{cm}}
\safemath{\m}{\,\mathrm{m}}
\safemath{\W}{\,\mathrm{W}}
\safemath{\mW}{\, \mathrm{mW}}
\safemath{\J}{\,\mathrm{J}}
\safemath{\K}{\,\mathrm{K}}
\safemath{\bit}{\,\mathrm{bit}}
\safemath{\nat}{\,\mathrm{nat}}

\safemath{\define}{\triangleq}					% definition

\safemath{\equivalent}{\sim}
\safemath{\distas}{\sim}					% distributed according to
\safemath{\sdiff}{\Delta}				% symmetric set difference
\safemath{\setdiff}{\setminus}				% set difference

\safemath{\reals}{\mathbb R}
\safemath{\positivereals}{\reals_{+}}
\safemath{\integers}{\mathbb Z}
\safemath{\posint}{\integers_{+}}
\safemath{\naturals}{\mathbb N}
\safemath{\posnaturals}{\naturals_{+}}
\safemath{\complexset}{\mathbb C}
\safemath{\rationals}{\mathbb Q}

\newcommand{\natseg}[2]{[#1 \text{\phantom{\tiny{.}}:\phantom{\tiny{.}}} #2]} % segment of natural numbers
\newcommand{\matseg}[3]{\left[#1\right]_{ #2 , #3 }} % entry of a matrix or a vector
\newcommand{\vecseg}[2]{#1_{ #2}} % entry of a matrix or a vector
\newcommand*{\fancyrefparlabelprefix}{par}		% Part
\newcommand*{\fancyrefchalabelprefix}{cha}		% Chapter
\newcommand*{\fancyrefapplabelprefix}{app}		% Appendix
\newcommand*{\fancyrefthmlabelprefix}{thm}		% Theorem
\newcommand*{\fancyreflemlabelprefix}{lem}		% Lemma
\newcommand*{\fancyrefcorlabelprefix}{cor}		% Corolary
\newcommand*{\fancyrefdeflabelprefix}{dfn}		% Definition

\frefformat{vario}{\fancyrefparlabelprefix}{Part~#1}
\frefformat{vario}{\fancyrefchalabelprefix}{Chapter~#1}
\frefformat{vario}{\fancyrefseclabelprefix}{Section~#1}
\frefformat{vario}{\fancyrefthmlabelprefix}{Theorem~#1}
\frefformat{vario}{\fancyreflemlabelprefix}{Lemma~#1}
\frefformat{vario}{\fancyrefcorlabelprefix}{Corollary~#1}
\frefformat{vario}{\fancyrefdeflabelprefix}{Definition~#1}
\frefformat{vario}{\fancyreffiglabelprefix}{Figure~#1}
\frefformat{vario}{\fancyrefapplabelprefix}{Appendix~#1}
\frefformat{vario}{\fancyrefeqlabelprefix}{(#1)}

\safemath{\iSet}{\setI}
\safemath{\rel}{\bowtie}					% relation
\safemath{\eqrel}{\sim}					% equivalence relation
\safemath{\rlord}{\leq}					% reflexive linear ordering
\safemath{\slord}{<}						% strict linear ordering
\safemath{\rpord}{\preceq}				% reflexive partial ordering
\safemath{\rrpord}{\succeq}				% reversed reflexive partial ordering
\safemath{\spord}{\prec}					% strict partial ordering
\safemath{\sig}{\sigma}					% sigma-{algebra, ring,...}
\safemath{\metric}{d}					% metric
\safemath{\setfun}{\Phi}					% set function
\safemath{\measure}{\mu}					% measure
\safemath{\altmeasure}{\lambda}					% measure
\newcommand{\outerm}[1]{#1^{\star}}		% marks outer measures.
\newcommand{\innerm}[1]{#1_{\star}}		% marks inner measures.
\safemath{\omeasure}{\outerm{\measure}}		% outer measure
\safemath{\imeasure}{\innerm{\measure}}		% inner measure	
\safemath{\aecol}{\colS^{\star}_{\measure}} % collection of almost equal sets
\safemath{\emeasure}{\bar{\measure}_{0}}	% measure extension
\safemath{\rmeasure}{\tilde{\measure}}	% restricted measure
\safemath{\bmeasure}{\measure_{0}}		% basic measure on a semiring
\safemath{\glength}{\measure_{\altfun}}	% generalized length
\safemath{\lebmea}{\lambda}				% Lebesgue length and measure
\safemath{\blebmea}{\lebmea_{0}}			% pre-lebesgue-measure
\safemath{\sfun}{s}						% simple function
\safemath{\absintspace}{\colL^{1}}		% space of abs. integrable functions
\safemath{\sqintspace}{\colL^{2}}		% space of square integrable functions
\safemath{\abssumspace}{l^{1}}		% space of abs. summable sequences
\safemath{\sqsumspace}{l^{2}}		% space of square summable sequences

\safemath{\sfield}{\setF}				% scalar field
\safemath{\vectors}{\setV}				% set of vectors
\safemath{\vecspace}{(\vectors,\sfield)}	% vector space
\safemath{\linspace}{\setV}				% linear space
\safemath{\clinspace}{(\linspace,\sfield)} % linear space
\safemath{\nspace}{\setU}				% normed space
\safemath{\metspace}{\setM}				% metric space
\safemath{\bspace}{\setB}				% Banach space
\safemath{\ipspace}{\setG}				% inner product space
\safemath{\hilspace}{\setH}				% Hilbert space
\safemath{\blospace}{\setG}				% set of bounded linear oprators
\safemath{\lop}{\opT}					% linear operator
\safemath{\altlop}{\opS}					% alternative linear operator
\safemath{\nullsp}{\nullspace(\lop)}		% null space of the linear operator
\safemath{\lfun}{l}						% linear functional
\safemath{\altlfun}{g}					% alternative linear functional
\newcommand{\dual}[1]{#1^{'}}			% dual space
\safemath{\dsum}{\oplus}					% direct sum
\safemath{\funspace}{\colL}				% function space
\renewcommand{\adj}[1]{#1^{\times}}		% adjoint operator generator
\safemath{\adjlop}{\adj{\lop}}			% adjoint operator
\safemath{\hadjlop}{\hadj{\lop}}			% Hilbert adjoint operator
\safemath{\tow}{\xrightarrow{w}}			% weak convergence
\safemath{\tows}{\xrightarrow{w^{*}}}		% weak* convergence
\safemath{\cparam}{\lambda}
\safemath{\lopl}{\lop_{\cparam}}		
\safemath{\iop}{\opI}					% identity operator
\safemath{\resolop}{\opR}				% resolvent operator
\safemath{\resolvent}{\resolop_{\cparam}(\lop)}	% resolvent operator
\safemath{\reset}{\setQ}
\safemath{\spectrum}{\setS}
\safemath{\resolset}{\reset(\lop)}		% resolvent set
\safemath{\lopspec}{\spectrum(\lop)}		% spectrum of a linear operator
\safemath{\pspec}{\spectrum_{p}(\lop)}	% point spectrum
\safemath{\cspec}{\spectrum_{c}(\lop)}	% continuous spectrum
\safemath{\rspec}{\spectrum_{r}(\lop)}	% residual spectrum
\safemath{\ev}{\cparam}					% eigenvalue
\newcommand{\specrad}[1]{r_{#1}}			% spectral radius
\safemath{\lopsrad}{\specrad{\lop}}		% spectral radius
\safemath{\pop}{\opP}					% projection operator

\safemath{\specfam}{\colE}				% spectral family
\safemath{\specop}{\opE_{\cparam}}		% spectral projection operator
\safemath{\altspecop}{\opE_{\mu}}		% alternat spectral projection operator
\safemath{\vmulti}{\vecone}				% vector multiplicative identity
\safemath{\unitaryop}{\opU}				% unitary operator
\safemath{\sval}{\sigma}					% singular value
\safemath{\corrcoef}{\rho}				% canonical correlation coefficient
\safemath{\sangle}{\theta}				% angle between subspaces

\let\time\undefined
\safemath{\iset}{\setI}				% index set
\safemath{\shift}{\nu}
\safemath{\scale}{\alpha}
\safemath{\time}{t}
\safemath{\specfreq}{\theta}	
\newcommand{\transopgen}[1]{\opT_{#1}} % translation operator
\safemath{\transop}{\transopgen{\delay}}
\newcommand{\modopgen}[1]{\opM_{#1}}	% modulation operator
\safemath{\modop}{\modopgen{\shift}}
\newcommand{\dilopgen}[1]{\opD_{#1}}	% dilation operator
\safemath{\dilop}{\dilopgen{\scale}}
\safemath{\fram}{\setF}				% frame
\safemath{\dfram}{\dual{\fram}}		% dual frame
\safemath{\ufb}{B}					% upper frame bound
\safemath{\lfb}{A}					% lower frame bound
\safemath{\sop}{\hadj{\aop}}				% frame synthesis operator
\safemath{\aop}{\opT}			% frame analysis operator 		// modifide by christoph bunte
\safemath{\fop}{\opS}				% frame operator 			// modified by christoph bunte 
\safemath{\daop}{\tilde\opT}			% dual frame analysis operator 		// added by christoph bunte
\safemath{\dsop}{\hadj{\tilde\opT}}				% dual frame synthesis operator 			// added by christoph bunte 

\safemath{\ifop}{\inv{\fop}}			% inverse frame operator
\safemath{\rifop}{\fop^{-1/2}}			% square root of inverse frame operator
\safemath{\transeq}{\setT}			% sequence of translates
\safemath{\nfun}{\Phi}				% ``Nyquist series''
\safemath{\funvec}{\vecf}			%  function as a vector
\safemath{\altfunvec}{\vecg}

\safemath{\samplespace}{\Omega}
\safemath{\probspace}{(\samplespace,\sfield,\Prob)}	% probability space
\safemath{\ccoef}{\rho}			% correlation coefficient
\safemath{\infstate}{\vecpi}				% steady state vector
\safemath{\typset}{\setA_{\epsilon}^{(N)}}	% typical set
\safemath{\expequal}{\doteq}				% equal to first order in the exponent
\safemath{\code}{C}						% code
\safemath{\dstringset}{\setD^{\star}}		% set of finite length D-ary strings
\safemath{\cwlength}{l}					% codeword length
\safemath{\elength}{L}					% expected codeword length
\safemath{\extension}{C^{\star}}			% code extension
\safemath{\approaches}{\rightarrow}		% i.e., x_{n} -> x
\safemath{\evnt}{\setA}					% event A
\safemath{\altevnt}{\setB}					% event B
\safemath{\rv}{\rndx}					% random variable X
\safemath{\altrv}{\rndy}					% random variable Y
\safemath{\complexrv}{\rndu}					% complex random variable U
\safemath{\altcrv}{\rndv}				% complex random variable V
\safemath{\rvec}{\rvecx}					% random vector X
\safemath{\altrvec}{\rvecy}				% random vector Y
\safemath{\crvec}{\rvecu}				% complex random vector U
\safemath{\altcrvec}{\rvecv}				% complex random vector V
\safemath{\variance}{\sigma^{2}}			% variance
\safemath{\map}{T}						% mapping
\safemath{\jacobian}{J}					% jacobian
\safemath{\wvec}{\rvecw}					% white random vector
\safemath{\wrv}{\rndw}					% white noise process
\safemath{\orthmat}{\matQ}				% orthogonal matrix
\safemath{\evmat}{\matLambda}			% eigenvalue matrix (diagonal)
\safemath{\identity}{\matidentity}		% identity matrix
\safemath{\innovec}{\vecv}				% innovations vector
\safemath{\convas}{\xrightarrow{\text{a.s.}}}	% almost sure convergence
\safemath{\convr}{\xrightarrow{\text{r}}}	% convergence in r-th mean
\safemath{\convp}{\xrightarrow{\text{P}}}	% convergence in probability
\safemath{\convd}{\xrightarrow{\text{D}}}	% convergence in distribution
\safemath{\ltis}{\opL}				% LTI system
\safemath{\ir}{h}					% impulse response
\safemath{\tf}{\MakeUppercase{\ir}}	% transfer function
\newtheorem{thm}{Theorem}
\newtheorem{lem}[thm]{Lemma}
\newtheorem{rem}{Remark}

\newcommand{\blocklength}{T}	 			% block length
\newcommand{\rankcov}{Q}	 			    % rank of correlation matrix
\newcommand{\vinp}{\vecx}	 			    % input vector
\newcommand{\vinpinv}{\vecz}	 			    % inverse input vector
\newcommand{\vinpc}{x}	 			    % components of the input vector
\newcommand{\vinpinvc}{z}	 			    % components of the input vector
\newcommand{\invinpc}{\veczc}					% inverse of the input vector
\newcommand{\minp}{\matX}	 			    % input matrix
\newcommand{\vout}[1]{\vecy_{#1}}	 			    % output vector
\newcommand{\vectout}{\vecy}	 			% vectorized output of SIMO
\newcommand{\vnoise}[1]{\vecw_{#1}}	 			    % noise vector
\newcommand{\vectnoise}{\vecw}	 			    % noise matrix
\newcommand{\vectoutnn}{\hat\vecy}	 			    % vectorized output matrix without noise
\newcommand{\vectoutnnc}{\hat\vecyc}	 			    % vectorized output matrix without noise
\newcommand{\const}{\landauO(1)}	 			    % vectorized output matrix without noise
\newcommand{\constalt}{c}	 			    % vectorized output matrix without noise

\newcommand{\sqrtcov}{\matP}	 			% sqruare root of covariance matrix
\newcommand{\sqrtcovsub}{\tilde\matP}	 	% submatrix of sqruare root of covariance matrix
\newcommand{\sqrtcovc}{\matPc}	 			% components of sqruare root of covariance matrix
\newcommand{\vchannel}[1]{\vech_{#1}}	 			% vector of channel gains
\newcommand{\viid}[1]{\vecs_{#1}}	 			    % vector of iid gaussians
\newcommand{\miid}{\matS}	 			    % matrix of iid gaussians
\newcommand{\miidc}{\matSc}	 			    % components of matrix of iid gaussians
\newcommand{\vectiid}{\vecs}	 			% vectorized matrix of iid gaussians for SIMO
\newcommand{\vectiidc}{\vecsc}	 			% components vectorized matrix of iid gaussians for SIMO
\newcommand{\RXant}{Q}	 			    % number of receive antennas
\newcommand{\RXante}{\rankcov}	 			    % number of receive antennas=Q
\newcommand{\allel}{\diamond}						%place-holder for "all elements in the set"

\safemath{\propspark}{\text{Property (A)}}

\begin{document}
\IEEEoverridecommandlockouts

%\title{A lower bound on the high-SNR capacity of a single-input multiple-output fading channel}
\title{The SIMO Pre-Log Can Be Larger Than the SISO Pre-Log}

%\svnInfo $Id: dmb_isit09.tex 2619 2009-04-28 10:00:27Z gdurisi $ 

% add a PDFinfo field to store metadata in the output PDF
% \pdfinfo{
% 	/Title		(On the Sensitivity of  Noncoherent Capacity to the Channel Model)
% 	/Author		(Giuseppe Durisi, Veniamin I. Morgenshtern, Helmut Boelcskei)
% 	/Subject		(SVN revision \svnInfoRevision, \today)
% 	/Keywords	(ISIT paper)
% }

\author{
\IEEEauthorblockN{Veniamin I. Morgenshtern, Giuseppe Durisi, and  Helmut B\"{o}lcskei }
\IEEEauthorblockA{ETH Zurich, 8092 Zurich, Switzerland\\
E-mail: \{vmorgens, gdurisi, boelcskei\}@nari.ee.ethz.ch\\} 
%\thanks{This work was supported...}
}

% make the title area
\maketitle

%%%%%%%%%%%%%%%%
\begin{abstract}
We establish a lower bound on the noncoherent capacity pre-log of a temporally correlated Rayleigh block-fading single-input multiple-output (SIMO) channel. 
Surprisingly, when the covariance matrix of the channel satisfies a certain technical condition related to the cardinality of its smallest set of linearly dependent rows, this lower bound reveals that the capacity pre-log in the SIMO case is larger than that in the single-input single-output (SISO) case.

\end{abstract}
\IEEEpeerreviewmaketitle
%%%%%%%%%%%%%%%%%%%%%%%%%%%%%%%%%%%%%%%%%%%%%
\section{Introduction}
\label{sec:introduction}

It is well known that the \emph{coherent-capacity} pre-log  
(i.e., the asymptotic ratio between capacity and the logarithm of SNR, as SNR goes to infinity) of a single-input multiple-output (SIMO) fading channel is equal to $1$ and is, hence, the same as  that of a single-input single-output (SISO) fading channel~\cite{telatar99-11}. 
In the practically more relevant \emph{noncoherent setting}, where neither transmitter nor receiver have channel-state information, but both are aware of the channel statistics, the effect of multiple antennas on the capacity\footnote{In the remainder of the paper, we consider the noncoherent setting only. Consequently, we will refer to capacity in the noncoherent setting simply as capacity. 
Furthermore, we shall assume Rayleigh fading throughout.} pre-log is understood only for a specific simple channel model, namely, the \emph{constant block-fading} model.
In this model, the channel is assumed to remain constant over a block of $\blocklength$ symbols and to change in an independent fashion from block to block~\cite{marzetta99-01}. For this model, the SIMO capacity pre-log is again equal to the SISO capacity pre-log, but, differently from the coherent case, is given by~$1-1/\blocklength$~\cite{hochwald00-05,zheng02-02}. 

A more general way of capturing channel variations in time is to assume that the fading process is \emph{stationary}. 
In this case, the capacity pre-log is known only in the SISO~\cite{lapidoth05-02} and the MISO~\cite[Thm. 4.15]{koch09} cases. 
The capacity bounds for the SIMO stationary-fading channel available in the literature~\cite[Thm. 4.13]{koch09} do not allow one to determine whether the capacity pre-log in the SIMO case can be larger than that in the SISO case.

In this paper, we focus on a channel model that can be seen as lying in between the general stationary-fading model considered in~\cite{lapidoth05-02,koch09}, and the simpler constant block-fading model analyzed in~\cite{marzetta99-01,zheng02-02}. 
Specifically, we assume that the fading process is independent across blocks of length~$\blocklength$ and temporally correlated within blocks, with the rank of the corresponding $\blocklength\times\blocklength$ channel covariance matrix given
by\footnote{When $\rankcov=\blocklength$, capacity is known to grow double-logarithmically in SNR~\cite{lapidoth03-10}, and, hence, the capacity pre-log is zero.} $\rankcov<\blocklength$.
For this channel model, referred to as the \emph{correlated block-fading} model in the following, the SISO capacity pre-log  is equal to~$1-\rankcov/\blocklength$~\cite{liang04-12}.\footnote{The constant block-fading model is obviously a special case ($\rankcov=1$) of the correlated block-fading model.}
The SIMO and MIMO capacity pre-logs are not known in this case. 
% It is, however, conjectured in~\cite{liang04-12} that the capacity pre-log does not increase if multiple antennas are added at the receiver side.
A conjecture in~\cite{liang04-12} on the MIMO capacity pre-log implies that the capacity pre-log in the SIMO case would be the same as that in the SISO case.
% 
% multiple antennas at the receiver do not increase the capacity pre-log.
In this paper, we disprove the conjecture in~\cite{liang04-12} by showing that in the SIMO case a capacity pre-log of~$1-1/\blocklength$ can be obtained when the number of receive antennas is equal to~$\rankcov$, and the channel covariance matrix satisfies a certain technical condition detailed in Theorem~\ref{thm:mainLB}.

%--
\subsubsection*{Notation}

  Uppercase boldface letters denote matrices, and lowercase boldface letters designate vectors. 
The all-zero matrix of appropriate size is written as $\veczero$.
The element in the $i$th row and $j$th column of a  matrix $\mat$ is denoted as $\matc_{i,j}$, and the $i$th component of the vector~$\vectr$ is $\vectrc_i$. 
For a vector $\vectr$, $\diag(\vectr)$ denotes the diagonal matrix that has the entries of $\vectr$ on its main diagonal.
 The superscripts~$\tp{}$ and~$\herm{}$ stand for transposition and Hermitian transposition, respectively. The expectation operator is denoted as~$\Ex{}{\cdot}$. 
For two matrices~$\mat$ and~$\altmat$, we designate the Kronecker product as~$\mat \kron \altmat$; 
to simplify notation, we use the convention that the ordinary matrix product always precedes the Kronecker product, i.e., $\matA\matB\kron\matC=(\matA\matB)\kron\matC$.
%we assume that the matrix product has a higher priority than the Kronecker product and omit unnecessary brackets in the remainder of the paper.
For two functions~$\fun(x)$ and~$\altfun(x)$, the notation~$\fun(x)=\landauO(\altfun(x))$ means that~$\lim_{x\to \infty} \abs{\fun(x)}\!/\!\abs{\altfun(x)}$ is bounded above by a constant.
We use $\natseg{n}{m}$ to designate the set of natural numbers $\left\{n, n+1,\ldots,m\right\}$.
Let $\altfunvec(\vectr)$ be a vector-valued function; then ${\partial \altfunvec}/{\partial \vectr}$ denotes the Jacobian matrix of the function $\altfunvec(\vectr)$, i.e., the matrix that contains the partial derivative ${\partial \altfun_i}/{\partial \vectrc_j}$ in its $i$th row and $j$th column.
%The spark of a matrix~$\mat$, denoted as $\spark(\mat)$, is defined as the smallest number of linearly dependent columns of $\mat$.
%
We write $\card{\setI}$ to denote the cardinality of the set $\setI$.
For an $M\times N$ matrix $\mat$, and two sets of indices $\setI\subset \natseg{1}{M}$ and $\setJ\subset \natseg{1}{N}$, we use~$\matseg{\mat}{\setI}{\setJ}$ to denote the  $\card{\setI}\times \card{\setJ}$ submatrix of $\mat$ containing the elements
$\left[\matc_{i,j}\right]_{i\in \setI, j\in \setJ}$.	
Similarly, for an $N$-dimensional vector $\vectr$ and a set $\setI\subset\natseg{1}{N}$, we define $\vecseg{\vectr}{\setI}\define [\vectrc_i]_{i\in\setI}$. 
For an $M\times N$ matrix $\mat$, we set 
$\matseg{\mat}{\allel}{\setJ}\define\matseg{\mat}{\natseg{1}{M}}{\setJ}$ and $\matseg{\mat}{\setI}{\allel}\define\matseg{\mat}{\setI}{\natseg{1}{N}}$.
Furthermore,
\begin{equation}
	\label{eq:Ddef}
		\diagd{\mat}\define\begin{bmatrix}
		\diag(\tp{[\matc_{11}\dots \matc_{M 1}]})\\\vdots\\ \diag(\tp{[\matc_{1 N} \dots \matc_{M N}]}) 	
		\end{bmatrix}.
\end{equation}
% %
The eigenvalues of the $N\times N$ matrix $\mat$ are denoted by  $\eval_1(\mat)\ge\cdots\ge \eval_N(\mat)$. 
The logarithm to the base 2 is written as $\log(\cdot)$.
Finally, $\jpg(\vecm,\matC)$ stands for the distribution of a
jointly proper Gaussian (JPG) random vector with mean~$\vecm$ and covariance
matrix~$\matC$.
%%%%%%%%%%%%%%%%%%%%%%

\section{System Model} % (fold)
\label{sec:system_model}
We consider a SIMO channel with $\RXant$ receive antennas. 
The fading in each \emph{component channel} follows the correlated block-fading model described in the previous section, namely, it is independent across blocks of length~$\blocklength$, and correlated within blocks, with the rank of the corresponding channel covariance matrix given by $\rankcov<\blocklength$. 
Note that we assume the rank of the channel covariance matrix to be equal to the number of receive antennas. 
Our analysis relies heavily on this assumption. 
Across component channels, the fading is independent and identically distributed.
% The block size~$\blocklength$ and the temporal covariance matrix are the same for each component channel.
The input-output (I/O) relation (within any block) for the $m$th component channel can be written as
\begin{equation*}
	%\label{eq:model1}
	%\nonumber
	\vout{m}=\sqrt{\SNR}\,\diag(\vchannel{m}) \vinp+\vnoise{m}, \quad m\in\natseg{1}{\RXant}
\end{equation*}
where the vector $\vinp=\tp{[\vinpc_1 \cdots\, \vinpc_\blocklength]}\in \complexset^\blocklength$ contains the $\blocklength$-dimensional signal transmitted within the block, and the vectors $\vout{m},\vnoise{m}\in\complexset^\blocklength$ contain  the corresponding received signal and additive noise, respectively, at the $m$th antenna. 
Finally, $\vchannel{m} \in \complexset^\blocklength$ contains the channel coefficients between the transmit antenna and the $m$th receive antenna. 
We assume that~$\vnoise{m}\distas\jpg(\veczero,\matI_{\blocklength})$ and~$\vchannel{m}\distas~\jpg(\veczero,\sqrtcov\herm\sqrtcov)$ are mutually independent (and independent across $m$) and that $\sqrtcov\in\complexset^{\blocklength\times\rankcov}$ (which is the same for all blocks) has rank~$\rankcov<\blocklength$.
It will turn out convenient to write the channel-coefficient vector in whitened form as $\vchannel{m}=\sqrtcov \viid{m}$, where $\viid{m}\distas\jpg(\veczero,\matI_{\rankcov})$.
Finally, we assume that $\viid{m}$ and $\vnoise{m}$ change in an independent fashion from block to block.%, and that coding is performed over infinitely many independent blocks.

If we define $\tp\vectout\define [\tp{\vout{1}} \cdots\, \tp{\vout{\RXant}}]$,
$\tp\vectiid\define [\tp{\viid{1}} \cdots\, \tp{\viid{\RXant}}]$, $\tp\vectnoise\define [\tp{\vnoise{1}} \cdots \tp{\vnoise{\RXant}}]$, and $\minp\define\diag(\vinp)$, 
we can write the channel I/O relation in the following---more compact---form%, both of which will be useful:
\begin{equation}
	\label{eq:IOstacked}
	\vectout=\sqrt{\SNR}\left(\identity_{\RXant}\kron\minp\sqrtcov\right)\vectiid+\vectnoise.
\end{equation}
%
%
%which will be convenient for the ensuing analysis.
% Under the assumption of the input $\vinp$ being subject to an average-power constraint according to
% \begin{equation}
% 	\label{eq:powerconstraint}
% 	\Ex{}{\vecnorm{\vinp}^2}\le \blocklength
% \end{equation} 
The capacity of the channel~\eqref{eq:IOstacked} is defined as
\begin{equation}
	\label{eq:capacitydef}
\capacity(\SNR)\define(1/\blocklength)\sup_{\pdf{\vinp}(\cdot)} \mi(\vinp;\vectout)
\end{equation}
where the supremum is taken over all input distributions $\pdf{\vinp}(\cdot)$ that satisfy the average-power constraint $\Ex{}{\vecnorm{\vinp}^2}\le \blocklength$.
%[TODO: what is the goal? What is presented next?]

\section{Intuitive Analysis}
\label{sec:intuition}
In this section, we describe a ``back-of-the-envelope'' method for guessing the capacity pre-log. A formal justification of this procedure is provided in Section~\ref{sec:flwSIMO}.

The capacity pre-log characterizes the asymptotic behavior of the fading-channel capacity at high SNR, i.e., in the regime where the additive noise can ``effectively'' be ignored. 
In order to guess the capacity pre-log, we therefore consider the problem of identifying the transmit symbols $\vinpc_{i},\ i\in\natseg{1}{\blocklength},$ from the \emph{noise-free} (and rescaled) observation
\begin{equation}
	\label{eq:IOnonoise}
	\vectoutnn\define\left(\identity_{\RXant}\kron\minp\sqrtcov\right)\vectiid.
\end{equation}
Specifically, we shall ask the question: ``How many symbols $\vinpc_{i}$ can be identified uniquely from $\vectoutnn$ given that the channel coefficients $\vectiid$ are unknown but the statistics of the channel, i.e., the matrix $\sqrtcov$, are known?'' 
The claim we make is that the capacity pre-log is given by the number of these symbols divided by the block length $\blocklength$. 

We start by noting that the unknown variables in~\eqref{eq:IOnonoise} are $\vectiid$ and $\vinp$, which means that we have a quadratic system of equations. It turns out, however, that the simple change of variables $\invinpc_i\define 1/\vinpc_i,\ i\in\natseg{1}{\blocklength},$ (we make the technical assumption $0<\abs{\vinpc_i}<\infty,\ i\in\natseg{1}{\blocklength}$, in the remainder of this section) transforms~\eqref{eq:IOnonoise} into a system of equations that is linear in $\vectiid$ and $\invinpc_i,\ i\in\natseg{1}{\blocklength}$. Since the transformation $\invinpc_i\define 1/\vinpc_i$ is invertible for $0<\vinpc_i<\infty$, uniqueness of the solution of the linear system of equations in $\vectiid$ and $\invinpc_i,\ i\in\natseg{1}{\blocklength},$ is equivalent to uniqueness of the solution of the quadratic system of equations in $\vectiid$ and $\vinpc_i,\ i\in\natseg{1}{\blocklength}$. For simplicity of exposition and concreteness, we consider the special case $\blocklength=3$ and $\RXant=2$. A direct computation reveals that~\eqref{eq:IOnonoise} is equivalent to
\begin{equation}
	\label{eq:linearsystem1}
	\begin{bmatrix}
		\sqrtcovc_{11} & \sqrtcovc_{12} & 0 & 0 & \vectoutnnc_1 & 0 & 0 \\
		\sqrtcovc_{21} & \sqrtcovc_{22} & 0 & 0 & 0 & \vectoutnnc_2 & 0 \\
		\sqrtcovc_{31} & \sqrtcovc_{32} & 0 & 0 & 0 & 0 & \vectoutnnc_3 \\
		0 & 0 & \sqrtcovc_{11} & \sqrtcovc_{12} & \vectoutnnc_4 & 0 & 0 \\
		0 & 0 & \sqrtcovc_{21} & \sqrtcovc_{22} & 0 & \vectoutnnc_5 & 0 \\
		0 & 0 & \sqrtcovc_{31} & \sqrtcovc_{32} & 0 & 0 & \vectoutnnc_6 \\
	\end{bmatrix}
	\begin{bmatrix}
		\vectiidc_1\\
		\vectiidc_2\\
		\vectiidc_3\\
		\vectiidc_4\\
			-\invinpc_1\\
			-\invinpc_2\\
			-\invinpc_3\\
	\end{bmatrix}=\veczero.
\end{equation}
The solution of this linear system of equations is not unique, as we have 6 equations in 7 unknowns.
% This linear system of equations can never have a unique solution as we have 6 equations for 7 unknowns. 
The $\vinpc_i=1/\invinpc_i,\ i\in\natseg{1}{3},$ can, therefore, not be determined uniquely from $\vectoutnn$. 
However, if we transmit one pilot symbol and two data symbols, the system of equations becomes solvable. 
Take for example $\vinpc_1=1$ and let the receiver know the value of this (pilot) symbol.  Then~\eqref{eq:linearsystem1} reduces to the following inhomogeneous system of 6 equations in 6 unknowns
\begin{equation}
	\underbrace{\begin{bmatrix}
		\sqrtcovc_{11} & \sqrtcovc_{12} & 0 & 0 & 0 & 0\\
		\sqrtcovc_{21} & \sqrtcovc_{22} & 0 & 0 & \vectoutnnc_2 & 0\\
		\sqrtcovc_{31} & \sqrtcovc_{32} & 0 & 0 & 0 & \vectoutnnc_3\\
		0 & 0& \sqrtcovc_{11} & \sqrtcovc_{12}& 0 & 0 \\
		0 & 0& \sqrtcovc_{21} & \sqrtcovc_{22} & \vectoutnnc_5 & 0  \\
		0 & 0& \sqrtcovc_{31} & \sqrtcovc_{32} & 0 & \vectoutnnc_6 \\
	\end{bmatrix}}_\altmat
	\begin{bmatrix}
		\vectiidc_1\\
		\vectiidc_2\\
		\vectiidc_3\\
		\vectiidc_4\\
		-\invinpc_2\\
		-\invinpc_3\\
	\end{bmatrix}=
	\begin{bmatrix}
		\vectoutnnc_1\\
		0\\
		0\\
		\vectoutnnc_4\\
		0\\
		0\\
	\end{bmatrix}.
	\label{eq:inhomeq}
\end{equation}
This system of equations has a unique solution if $\det\altmat\ne 0$. % and the right-hand-side (RHS) of~\eqref{eq:inhomeq} is nonzero. 
We prove in Appendix~\ref{app:bijection} 
%(this follows, for example, from the result $\det\matG\ne 0$ for $\matG$ defined in~\eqref{eq:Gdef}, which is proven in Appendix~\ref{app:bijection}) 
that under the technical condition on $\sqrtcov$ specified in Theorem~\ref{thm:mainLB} below, we, indeed, have that $\det\altmat\ne 0$ for almost all\footnote{Except for a set of measure zero.} $\vectoutnnc_2, \vectoutnnc_3, \vectoutnnc_5, \vectoutnnc_6$. %, and that the vector on the  RHS of~\eqref{eq:inhomeq} is nonzero for almost all $\vectoutnnc_1, \vectoutnnc_4$.  
It, therefore, follows that for almost all $\vectoutnn$, the system of equations~\eqref{eq:inhomeq} has a unique solution. 
Consequently, we can recover $\invinpc_2$ and $\invinpc_3$, and, hence, $\vinpc_2=1/\invinpc_2$ and $\vinpc_3=1/\invinpc_3$. 
Summarizing our findings, we expect that the capacity pre-log of the channel~\eqref{eq:IOstacked}, for the special case~$\blocklength=3$ and $\rankcov=2$, is equal to $2/3$. 
This is larger than the capacity pre-log of the corresponding SISO channel (i.e., the capacity pre-log of one of the component channels), which is  equal to $1-\rankcov/\blocklength=1/3$~\cite{liang04-12}. 

In general, we expect that under some technical conditions on $\sqrtcov$ the capacity pre-log of the SIMO channel as defined in Section~\ref{sec:system_model} % described in Section~\ref{sec:system_model} 
is equal to $(\blocklength-1)/\blocklength=1-1/\blocklength$. This is exactly what we intend to show rigorously in the next section.

\section{A Lower Bound on the Capacity Pre-Log} % (fold)
\label{sec:flwSIMO}
The main result of this paper is the following theorem.
\begin{thm}
\label{thm:mainLB}
Assume that	there exists a subset of indices $\setI\subset\natseg{1}{\blocklength}$ of cardinality $\card{\setI}=\rankcov+1$ such that the $[(\rankcov+1)\times\rankcov]$-dimensional submatrix $\sqrtcovsub\define \matseg{\sqrtcov}{\setI}{\allel}$ of the matrix $\sqrtcov$ in~\eqref{eq:IOstacked} satisfies the following \emph{\propspark:} Any set of~$\rankcov$ rows of~$\sqrtcovsub$ is linearly independent.
%
%  \begin{equation}
%  	 		\rank\lefto(\matseg{\sqrtcovsub}{\setO}{\allel}\right)=\rankcov
%  	 		\label{eq:sparkprop}
%  	\end{equation}
% for any set $\setO\subset\natseg{1}{\rankcov+1}$ of cardinality $\rankcov$.
%	%
	%  \begin{equation}
	% % 	 		\spark\lefto(\tp\sqrtcovsub\right)=\rankcov+1.
	%  	 		\label{eq:sparkprop}
	%  	\end{equation}
	%
%	\item the number of receive antennas in the channel I/O relation~\eqref{eq:IOstacked} is~$\rankcov$.
%\end{itemize}
%
%
Then, the capacity of the SIMO channel~\eqref{eq:IOstacked} can be lower-bounded as 
\begin{equation}
 		%\nonumber
\label{eq:mainbound}
 		\capacity(\SNR)\ge \left(1-1/\blocklength\right)\log\lefto(\SNR\right)+\const,\ \SNR\to\infty.
\end{equation}
% 	Assume that\footnote{In the rest of the paper we focus exclusively on the special case  $\rankcov=\RXant$ without explicitly mentioning it. The general case is still an open problem.} $\rankcov=\RXant$ and 
% 	that there is a submatrix $\sqrtcovsub$ of dimension $(\rankcov+1)\times\rankcov$ of the matrix $\sqrtcov$ in~\eqref{eq:IOstacked} that satisfies the following property
% 	% set of indexes $\setI\subset\natseg{1}{\blocklength}$ of size $\card{\setI}=\rankcov+1$, such that 
% 	%the corresponding submatrix $\matseg{\sqrtcov}{\setI}{}$ of matrix $\sqrtcov$ in~\eqref{eq:IOstacked} has $
% 	\begin{equation}
% 		\spark\lefto(\tp\sqrtcovsub\right)=\rankcov+1.
% 		\label{eq:sparkprop}
% 	\end{equation}
% 	Then the capacity of the SIMO channel \eqref{eq:IOstacked} can be lower-bounded as
% 	\begin{equation}
% 		\nonumber
% 		\capacity(\SNR)\ge \left(1-\frac{1}{\blocklength}\right)\log\lefto(\SNR\right)+\const,\ \SNR\to\infty.
% 	\end{equation}
% %	where here and in the rest of the paper $\const$ means a finite constant independent of $\SNR$.
\end{thm}
\begin{rem}
\label{rem:first}
For the special case $\blocklength=\rankcov+1$,~\eqref{eq:mainbound} yields a lower bound on the capacity pre-log that is tight. 
A matching upper bound can be obtained through steps similar to those in the proof of~\cite[Prop. 4]{liang04-12}. 
Establishing tight upper bounds on the capacity pre-log for general values of $\blocklength$, however, seems to be an open problem.
Different tools than those used in~\cite{liang04-12} are probably needed. 
\end{rem}
\begin{rem}
\label{rem:third}
When $\rankcov=1$, the channel in~\eqref{eq:IOstacked} reduces to a SISO constant block-fading channel, and the lower bound~\eqref{eq:mainbound} yields the correct capacity pre-log of~\cite{zheng02-02, hochwald00-05}. 	
\end{rem}
\begin{rem}
\label{rem:second}
	\propspark is not very restrictive and is satisfied by many practically relevant matrices $\sqrtcov$. 
	For example, removing any set of $\blocklength-\rankcov$ columns from a $\blocklength\times\blocklength$ discrete Fourier transform (DFT) matrix, results in a matrix that satisfies \propspark when $\blocklength$ is prime~\cite{tao05}. 
	%DFT matrices are often used to model the  temporal covariance of fading channels~\cite{liang04-12}.  
DFT covariance matrices occur naturally in basis-expansion models for time-varying channels~\cite{liang04-12}.
\end{rem}
\begin{proof}
We choose an input distribution for which the entries~$\vinpc_i,\   i\in\natseg{1}{\blocklength}$,  of~$\vinp$, are independent and identically distributed (\iid), have zero mean and unit variance,  and satisfy~$\Ex{}{\log \abs{\vinpc_i}}>-\infty$ and $\diffent(\vinpc_i)>-\infty$. For example, we can take $\vinpc_i\distas\jpg(0,1)$. 
% [Note that the class of input distributions that allows us to prove the theorem is surprisingly wide. This is because we are merely building a lower bound with a supposedly tight pre-log, but we do not try to find the correct constants in the high SNR capacity expansion.]
We then lower-bound $\mi(\vinp;\vectout)$ in~\eqref{eq:capacitydef}, evaluated for this input distribution. 
More precisely, we use $\mi(\vinp;\vectout)=\diffent(\vectout)-\diffent(\vectout\given \vinp)$ and bound the two differential entropy terms separately. Note that the class of input distributions for which~\eqref{eq:mainbound} holds is large.
This does not come as a surprise, as we are interested in the capacity pre-log only.

As~$\vectout$ conditional on~$\vinp$ is JPG, the conditional differential entropy~$\diffent(\vectout\given \vinp)$ can be upper-bounded in a straightforward fashion as follows:
\begin{align}
	\label{eq:UBcond}
		&\diffent(\vectout\given \vinp)=\RXante \blocklength \log(\pi e)\nonumber\\
		&+\Ex{\vinp}{\log\det \lefto(\identity_{\RXante \blocklength}+
		\SNR\left(\identity_{\RXante}\kron\minp\sqrtcov\right)\Ex{\vectiid}{\vectiid\herm\vectiid}\left(\identity_{\RXante}\kron\herm\sqrtcov\herm\minp\right)\right)}\nonumber	\\
		% % &=\RXante \blocklength \log(\pi e)+\Ex{\vinp}{\log\det \lefto(\npower\identity_{\RXante \blocklength}+
		% 		% 		\spower\left(\identity_{\RXante}\kron\minp\sqrtcov\right)\left(\identity_{\RXante}\kron\herm\sqrtcov\herm\minp\right)\right)}\\
		% 		% 		&=\RXante \blocklength \log(\pi e)+\Ex{\vinp}{\log\det \lefto(\npower\identity_{\RXante \blocklength}+
		% 		% 		\spower\left(\identity_{\RXante}\kron\minp\sqrtcov\herm\sqrtcov\herm\minp\right)\right)}\\
		% 		% 		&=\RXante \blocklength \log(\pi e)+\RXante\Ex{\vinp}{\log\det \lefto(\npower\identity_{ \blocklength}+
		% 		% 		\spower\left(\minp\sqrtcov\herm\sqrtcov\herm\minp\right)\right)}\\
		% 		&=\RXante \blocklength \log(\pi e)+\RXante\Ex{\vinp}{\log\det \lefto(\identity_{ \rankcov}+
		% 		\SNR\!\left(\herm\sqrtcov\herm\minp\minp\sqrtcov\right)\right)}\nonumber\\
		% 		&\stackrel{(a)}{\le} \RXante \blocklength \log(\pi e)+\RXante\log\det \lefto(\identity_{ \rankcov}+
		% 			\SNR\!\left(\herm\sqrtcov\Ex{\vinp}{\herm\minp\minp}\!\sqrtcov\right)\right)\nonumber\\
		% 		&=\RXante \blocklength \log(\pi e)+\RXante\sum_{i=1}^{\rankcov}\log \lefto(1+
		% 				\SNR\eval_i\lefto(\herm\sqrtcov\sqrtcov\right)\right)\nonumber\\
		&\le\rankcov^2\log \lefto(\SNR\right)+\const
\end{align}
where %~(a) follows from Jensen's inequality, and 
the inequality holds because $\sqrtcov$ has rank~$\rankcov$. 
% as a consequence of \propspark and, therefore, $\eval_i\lefto(\herm\sqrtcov\sqrtcov\right)>0$ for all $i\in\natseg{1}{\rankcov}$.

Finding a tight lower bound on $\diffent(\vectout)$ is the main difficulty of the proof.
In fact, the differential entropy of $\vectout$ is often intractable even for simple input distributions. 
The main technical contribution of this paper is presented in Section~\ref{sec:lowerboundhy} below, where we show that if \propspark is satisfied and if the input distribution satisfies the conditions specified at the beginning of this proof, we have
\begin{equation}
		\label{eq:lbdeY}
%\nonumber
	\diffent(\vectout)	\ge (\blocklength-1+\rankcov^2)\log(\SNR)+\constalt
	\end{equation}
	where $\constalt$, here and in the remainder of the paper, stands for a constant\footnote{The value of this constant can change at each appearance.} that is independent of $\SNR$.
 	Combining~\eqref{eq:UBcond} and~\eqref{eq:lbdeY} then yields the desired result.
	Note that in order to establish~\eqref{eq:UBcond} it is sufficient to use that $\sqrtcov$ has rank~$\rankcov$, whereas the more restrictive \propspark is crucial to establish~\eqref{eq:lbdeY}.
\end{proof}	

\subsection{A Lower Bound on $\diffent(\vectout)$} % (fold)
\label{sec:lowerboundhy}
The main idea of our approach
is to relate $\diffent(\vectout)$ to $\diffent(\vectiid,\vinp)= \diffent(\vectiid)+\diffent(\vinp)$, which is generally much simpler to compute than~$\diffent(\vectout)$.
It is possible to relate the entropies of two random vectors in a simple way if 
the vectors are of the same dimension and 
are connected by a \emph{deterministic one-to-one} (in the sense of~\cite[p.7]{rudin87}) function. %(see Lemma~\ref{lem:Entrchange} below). 
 %It is well known that  
%if two random vectors are connected  by a coordinate change, i.e. by a \emph{deterministic bijection} with a \emph{square} Jacobian matrix, then the entropies of these vectors are related in a simple way through the Jacobian of the coordinate change  
% It is clear that $\vectout$ and $(\vectiid,\vecx)$ are not connected by a deterministic one-to-one function 
% because $\vectout$ depends on the noise $\vectnoise$ whereas $(\vectiid,\vecx)$ does not. 
%One can readily see from \eqref{eq:IOstacked} that this is not the case for~$\vectout$ and $(\vectiid,\vecx)$.
This is not the case for~$\vectout$ and $\tp{[\tp\vectiid\ \tp\vecx]}$.
%, because the dimensionalities of $\vectout$ and of $\tp{[\tp\vectiid\ \tp\vecx]}$ are, in general, different; and, because $\vectout$ depends on the random noise vector $\vectnoise$ whereas $\tp{[\tp\vectiid\ \tp\vecx]}$ does not.
It is, however, possible to show that if $\vinpc_1$ is a fixed parameter, then there is a deterministic one-to-one function between $\tp{[\tp\vectiid\ \vecseg{\tp\vinp}{\natseg{2}{\blocklength}}]}$ and a specific \emph{subset} $\setJ\subset\natseg{1}{\blocklength\RXant}$ \emph{of components} of the \emph{noiseless} version~$\vectoutnn$ of the output vector.
% \begin{equation}
% 	\nonumber
%	$\vectoutnn\define \left(\identity_{\RXante}\kron\minp\sqrtcov\right)\vectiid,$
% \end{equation}
%We denote such subset of components as $\vecseg{\vectoutnn}{\setJ}$.
 This allows us to relate $\diffent(\vecseg{\vectoutnn}{\setJ}\given \vinpc_1)$ to $\diffent(\vectiid, \vecseg{\vinp}{\natseg{2}{\blocklength}} \given \vinpc_1)=\diffent(\vectiid, \vecseg{\vinp}{\natseg{2}{\blocklength}})$, which will turn out to be sufficient for our purposes as
 $\diffent(\vectout)$ can be linked to  $\diffent(\vecseg{\vectoutnn}{\setJ}\given \vinpc_1)$ according to~\eqref{eq:outputremlast} below. 
%Summarizing, we can relate $\diffent(\vectout)$ to $\diffent(\vecseg{\vectoutnn}{\setJ}\given \vinpc_1)$ and then to $\diffent(\vectiid, \vecseg{\vinp}{\natseg{2}{\blocklength}} \given \vinpc_1)$, which is simple to compute. 
We now describe the details of the proof program outlined above.

% This is detailed in Lemma~\ref{lem:bijection}. 

\begin{lem}
	\label{lem:bijection}
 	Assume that the matrix $\sqrtcov$ satisfies the conditions of Theorem~\ref{thm:mainLB} and take the submatrix $\sqrtcovsub$ defined in Theorem~\ref{thm:mainLB} to consist of the first~$\rankcov+1$ rows of  $\sqrtcov$ for simplicity.\footnote{This assumption will be made in the remainder of the paper, without explicitly mentioning it again.}
Let
	\begin{align}
		\label{eq:setJ}
		\setJ\define\natseg{1}{\blocklength}&\union \natseg{\blocklength+1}{\blocklength+\rankcov+1}\nonumber\\
		&\union \natseg{2\blocklength+1}{2\blocklength+\rankcov+1}
		%\setdiff \setI+\blocklength\right) \union \left(\natseg{1}{\blocklength}\setdiff \setI+2\blocklength\right)
		\union\cdots\nonumber\\
		&\union \natseg{(\RXante-1)\blocklength+1}{(\RXante-1)\blocklength+\rankcov+1}
	\end{align}
%	In words, the corresponding vector $\vecseg{\vectout}{\setJ}$ contains all the entries of~$\vout{1}$, plus the entries in the first~$\rankcov+1$ positions of the other output vectors~$\vout{m}$, $m\in\natseg{2}{\rankcov}$.
	where $\card{\setJ}=\blocklength-1+\rankcov^2$, and consider the vector-valued function $\vecseg{\vectoutnn}{\setJ}:\complexset^{\blocklength-1+\rankcov^2}\to\complexset^{\blocklength-1+\rankcov^2}$
	\begin{equation}
		\label{eq:maintransformation}
		\vecseg{\vectoutnn}{\setJ}\lefto(\vectiid,\vecseg{\vinp}{\natseg{2}{\blocklength}}\right)=\vecseg{\left(\left(\identity_{\RXante}\kron\minp\sqrtcov\right)\vectiid\right)}{\setJ}
	\end{equation}
parametrized by $\vinpc_1\ne 0$. To simplify the notation we will not indicate this parametrization explicitly.
%	where it is understood that $\lefto(\vectiid,\vecseg{\vinp}{\natseg{2}{\blocklength}}\right)$ are the variables of the function $\vecseg{\vectoutnn}{\setJ}(\cdot)$,  while $\vinpc_1\ne 0$ is fixed.
	The function $\vecseg{\vectoutnn}{\setJ}(\cdot)$ is one-to-one almost everywhere (a.e.) on $\complexset^{\blocklength-1+\rankcov^2}$.
\end{lem}
\begin{proof}
	See Appendix~\ref{app:bijection}.
\end{proof}
%The proof of this Lemma is straightforward but tedious and we present it in Appendix~\ref{app:bijection}. 
The following comments on Lemma~\ref{lem:bijection} are in order.
For $\blocklength=3$ and $\rankcov=2$ as in the simple example in Section~\ref{sec:intuition}, $\setJ=\natseg{1}{6}$, so that $\vecseg{\vectoutnn}{\setJ}=\vectoutnn$. 
Therefore, the 
one-to-one correspondence established in this lemma simply means that~\eqref{eq:IOnonoise} has a unique solution for fixed $\vinpc_1\ne 0$.
For the proof of the lemma, it is crucial that $\vinpc_1\ne 0$ is fixed. 
In fact, one can check that if none of the components of $\vinp$ is fixed, the resulting equivalent of the function $\vecseg{\vectoutnn}{\setJ}(\cdot)$ cannot be one-to-one, no matter how the set $\setJ$ is chosen.%\footnote{This follows from geometric considerations beyond the scope of this paper.}
Fixing $\vinpc_1\ne 0$ in order to make the function $\vecseg{\vectoutnn}{\setJ}(\cdot)$ be one-to-one corresponds to transmitting a pilot, as done in the simple example in Section~\ref{sec:intuition} by setting~$\vinpc_1=1$.
%The size of the set $\setJ$, which determines the pre-log lower bound, as we shall see below, is dictated by the requirement that the domain and the range of $\vecseg{\vectoutnn}{\setJ}(\cdot)$ must be of the same dimension. Once we fix $\vinpc_1$, the domain of the function $\vecseg{\vectoutnn}{\setJ}(\cdot)$ becomes $\blocklength-1+\rankcov^2$ dimensional and therefore $\setJ$ must contain $\blocklength-1+\rankcov^2$ elements. 
The cardinality of the set $\setJ$, which determines the lower bound on the capacity pre-log, as we shall see below, is dictated by the requirement that $\vecseg{\vectoutnn}{\setJ}$ and $\tp{\big[\tp\vectiid\ \vecseg{\tp\vinp}{\natseg{2}{\blocklength}}\big]}$ are of the same dimension, which implies that $\setJ$ must contain $\blocklength-1+\rankcov^2$ elements. 
The specific choice of \setJ in~\eqref{eq:setJ} simplifies the proof of the lemma (see Appendix~\ref{app:factorization}).

Lemma~\ref{lem:bijection} can be used to relate the conditional differential entropy $\diffent(\vecseg{\vectoutnn}{\setJ}\given \vinpc_1)$ to $\diffent(\vectiid, \vecseg{\vinp}{\natseg{2}{\blocklength}})$. Before doing so, we establish a simple lower bound on $\diffent\lefto(\vectout\right)$ that is explicit in $\diffent\lefto(\vecseg{\vectoutnn}{\setJ}\given \vinpc_1\right)$.
Let $\setN$ be the complement of $\setJ$ in $\natseg{1}{\RXante\blocklength}$.
	Then 
	\begin{align}
		\diffent\lefto(\vectout\right)
		&=\diffent\lefto(\vecseg{\vectout}{\setJ},\vecseg{\vectout}{\setN}\right)=\diffent\lefto(\vecseg{\vectout}{\setJ}\right)+\diffent\lefto(\vecseg{\vectout}{\setN}\given \vecseg{\vectout}{\setJ}\right)\nonumber\\
		&\ge\diffent\lefto(\sqrt{\SNR}\vecseg{\vectoutnn}{\setJ}+\vecseg{\vectnoise}{\setJ}\right)+\diffent\lefto(\vecseg{\vectout}{\setN}\given \vectiid, \vinp,\vecseg{\vectout}{\setJ}\right)\nonumber\\
		&\ge\diffent\lefto(\sqrt{\SNR}\vecseg{\vectoutnn}{\setJ}+\vecseg{\vectnoise}{\setJ}\given \vecseg{\vectnoise}{\setJ}\right)+\diffent\lefto(\vecseg{\vectnoise}{\setN}\right)\nonumber\\
		&%=\diffent\lefto(\sqrt{\SNR}\vecseg{\vectoutnn}{\setJ}\right)+\constalt
		\ge\card{\setJ}\log(\SNR)+\diffent\lefto(\vecseg{\vectoutnn}{\setJ}\given \vinpc_1\right)+\constalt.
		\label{eq:outputremlast}
	\end{align}
	% Notice that this chain of inequalities allows us to reduce the problem to considering the entropy of the \emph{noiseless} channel output $\vectoutnn$. 
	Through this chain of inequalities, we got rid of the noise~\vectnoise. 
	This corresponds to considering the noise-free I/O relation~\eqref{eq:IOnonoise} in the intuitive explanation given  in Section~\ref{sec:intuition}. 
	% This is precisely the reason why we could restrict ourselves to considering the noise-free channel~\eqref{eq:IOnonoise} in the intuition we gave in Section~\ref{sec:intuition}. 
	Inserting~$\card{\setJ}=\blocklength-1+\rankcov^2$ into~\eqref{eq:outputremlast}, we obtain the desired result~\eqref{eq:lbdeY} provided that $\diffent\lefto(\vecseg{\vectoutnn}{\setJ}\given \vinpc_1\right)$ is finite, which will be proved by means of the following lemma. 

\begin{lem}[Transformation of differential entropy]
	\label{lem:Entrchange}
	Assume that $\altfunvec:\complexset^N\to\complexset^N$ is a continuous vector-valued function that is one-to-one a.e. on $\complexset^N$. Let $\vectr\in\complexset^N$ be a random vector and $\altvectr=\altfunvec(\vectr)$.  
	Then
	\begin{equation*}
		%\label{eq:Entrchange}
		%\nonumber
		\diffent(\altvectr)=\diffent(\vectr)+\Ex{\vectr}{\log\abs{\det\lefto({\partial\altfunvec}/{\partial \vectr}\right)}}.
	\end{equation*}
\end{lem}
\begin{proof}
	The proof follows from the change-of-variable theorem for integrals~\cite[Thm. 7.26]{rudin87}.
	% and the properties of the $\log(\cdot)$ function in the definition of differential entropy.
\end{proof}

% The function $\vecseg{\vectoutnn}{\setJ}(\cdot)$ in~\eqref{eq:maintransformation} is continuous and is  one-to-one a.e. as shown in Lemma~\ref{lem:bijection}. Hence, we can apply Lemma~\ref{lem:Entrchange} to $\diffent(\vecseg{\vectoutnn}{\setJ}\given \vinpc_1=x)$ and write
% \begin{align}
% 	\label{eq:dec_change_var}
% 	&\diffent(\vecseg{\vectoutnn}{\setJ}\given \vinpc_1)=\int\!\! \pdf{\vinpc_1}(x) \diffent(\vecseg{\vectoutnn}{\setJ}\given \vinpc_1=x) dx\nonumber\\
% 	&\quad=\diffent(\vectiid, \vecseg{\vinp}{\natseg{2}{\blocklength}} \given \vinpc_1)+
% 	\Ex{\vectiid,\vinp}{\log\abs{\det{\frac{\partial \vecseg{\vectoutnn}{\setJ}}{\partial (\vectiid,\vecseg{\vinp}{\natseg{2}{\blocklength}})}}}}
% \end{align}
% where $\pdf{\vinpc_1}(x)$ is the density of $\vinpc_1$.
Let $\pdf{\vinpc_1}(x)$ denote the density of $\vinpc_1$. Then
\begin{align}
	\label{eq:dec_change_var}
	&\diffent(\vecseg{\vectoutnn}{\setJ}\given \vinpc_1)=\int\!\! \pdf{\vinpc_1}(x) \diffent(\vecseg{\vectoutnn}{\setJ}\given \vinpc_1=x) dx\nonumber\\
	&\quad=\diffent(\vectiid, \vecseg{\vinp}{\natseg{2}{\blocklength}} \given \vinpc_1)+
	\Ex{\vectiid,\vinp}{\log\abs{\det{\frac{\partial \vecseg{\vectoutnn}{\setJ}}{\partial (\vectiid,\vecseg{\vinp}{\natseg{2}{\blocklength}})}}}}
\end{align}
where in the second equality we applied Lemma~\ref{lem:Entrchange} to $\diffent(\vecseg{\vectoutnn}{\setJ}\given \vinpc_1=x)$, using that
the function $\vecseg{\vectoutnn}{\setJ}(\cdot)$ in~\eqref{eq:maintransformation} is continuous and is  one-to-one a.e. as shown in Lemma~\ref{lem:bijection}. 
The first term on the RHS of~\eqref{eq:dec_change_var} satisfies 
\begin{equation*}
	%\nonumber
	\diffent(\vectiid, \vecseg{\vinp}{\natseg{2}{\blocklength}} \given \vinpc_1)=\diffent(\vectiid)+\diffent(\vecseg{\vinp}{\natseg{2}{\blocklength}})>-\infty
\end{equation*}
where the inequality follows because the $\vinpc_i, i\in\natseg{2}{\blocklength},$ are \iid and have finite differential entropy.  % from the choice of the input distribution in the beginning of the proof of Theorem~\ref{thm:mainLB}.
It therefore remains to show that the second term on the RHS of~\eqref{eq:dec_change_var} is finite as well.
% \begin{equation}
% 	\label{eq:finitejacobian}
% 	\Ex{\vectiid,\vinp}{\log\abs{\det{\frac{\partial \vecseg{\vectoutnn}{\setJ}}{\partial (\vectiid,\vecseg{\vinp}{\natseg{2}{\blocklength}})}}}}>-\infty.
% \end{equation}
%
As the RHS of~\eqref{eq:maintransformation} is linear in $\vectiid$, we have that
\begin{equation*}
	%\nonumber
	%\frac{\partial \vecseg{\vectoutnn}{\setJ}}{\partial \vectiid}=\matseg{\identity_{\RXante}\kron\minp\sqrtcov}{\setJ}{\allel}.
	{\partial \vecseg{\vectoutnn}{\setJ}}/{\partial \vectiid}=\matseg{\identity_{\RXante}\kron\minp\sqrtcov}{\setJ}{\allel}.
\end{equation*}
Furthermore, using~\cite[Eq. (5), Sec. 7.2]{lutkepohl96}, the RHS of~\eqref{eq:maintransformation} can be rewritten as
% \begin{equation}
% 	\label{eq:equivform}
	$\vecseg{\vectoutnn}{\setJ}(\vectiid,\vecseg{\vinp}{\natseg{2}{\blocklength}})=\vecseg{\left((\miid\kron \identity_\blocklength)\diagd{\sqrtcov}\vinp\right)}{\setJ}$,
% \end{equation}
where the operator $\diagd{\cdot}$ was defined in~\eqref{eq:Ddef} and 
\begin{equation}
	\label{eq:miiddef}
	\tp{\miid}=[\viid{1}\dots\viid{\RXante}].
\end{equation}
Hence, we have that
\begin{equation*}
	%\nonumber
	%\frac{\partial \vecseg{\vectoutnn}{\setJ}}{\partial \vecseg{\vinp}{\natseg{2}{\blocklength}}}=\matseg{(\miid\kron \identity_\blocklength)\diagd{\sqrtcov}}{\setJ}{\natseg{2}{\blocklength}}.
	{\partial \vecseg{\vectoutnn}{\setJ}}/{\partial \vecseg{\vinp}{\natseg{2}{\blocklength}}}=\matseg{(\miid\kron \identity_\blocklength)\diagd{\sqrtcov}}{\setJ}{\natseg{2}{\blocklength}}.
\end{equation*}
To summarize, the Jacobian matrix in~\eqref{eq:dec_change_var} is given by 
\begin{equation}
	\label{eq:partialjacobian}
	\frac{\partial \vecseg{\vectoutnn}{\setJ}}{\partial (\vectiid,\vecseg{\vinp}{\natseg{2}{\blocklength}})}=\begin{bmatrix}
		\matseg{\identity_{\RXante}\kron\minp\sqrtcov}{\setJ}{\allel} & \matseg{(\miid\kron \identity_\blocklength)\diagd{\sqrtcov}}{\setJ}{\natseg{2}{\blocklength}}
	\end{bmatrix}.\\
\end{equation}
As shown in Appendix~\ref{app:factorization}, 
the determinant of this Jacobian matrix can be factorized as follows
\begin{align}
		\label{eq:partialjacobianfact}
		&\abs{\det{\begin{bmatrix}
			\matseg{\identity_{\RXante}\kron\minp\sqrtcov}{\setJ}{\allel} & \matseg{(\miid\kron \identity_\blocklength)\diagd{\sqrtcov}}{\setJ}{\natseg{2}{\blocklength}}
		\end{bmatrix}}}\nonumber\\
		&\quad=\prod_{j\in\natseg{\rankcov+2}{\blocklength}}\abs{\sum_{q=1}^{\rankcov} \miidc_{1q} \sqrtcovc_{j q}}
			\abs{\det\matM_1(\minp)}
		\abs{\det\matM_2(\miid)}\nonumber\\
		&\qquad\qquad\times
			\abs{\det\matM_3(\sqrtcov)}
			\abs{\det\matM_4(\miid)}
			\abs{\det\matM_5(\minp)}
	\end{align}
	where
	{\allowdisplaybreaks
	\begin{align*}
		&\matM_1(\minp)\define \identity_{\RXante}\kron\matseg{\minp}{\natseg{1}{\rankcov+1}}{\natseg{1}{\rankcov+1}}\\% \label{eq:M1}\\
		&\matM_2(\miid)\define \miid\kron \identity_{\rankcov+1}\\%\\\label{eq:M2}\\
		&\matM_3(\sqrtcov)\define \begin{bmatrix}
				\left(\identity_{\rankcov}\kron\sqrtcovsub\right) & \matseg{\diagd{\sqrtcovsub}}{\allel}{\natseg{2}{\rankcov+1}}
		\end{bmatrix}\\%\label{eq:M3}\\ %,\ 	\sqrtcovsub\define \matseg{\sqrtcov}{\setI}{\allel}\label{eq:M3}\\
		&\matM_4(\miid)\define 
		\begin{bmatrix}
			(\miid^{-1}\kron \identity_\rankcov) & \matzero \\ 
				\matzero &  	\identity_{{\rankcov}}
		\end{bmatrix}\\%\label{eq:M4}\\
		&\matM_5(\minp)\define \begin{bmatrix}
			\identity_{\rankcov^2} & \matzero \\ 
			\matzero &  \matseg{\minp}{\natseg{2}{\rankcov+1}}{\natseg{2}{\rankcov+1}}^{-1}
		\end{bmatrix}.%\label{eq:M5}
\end{align*}}%
Hence, we can rewrite the second term on the RHS of~\eqref{eq:dec_change_var} as
% Inserting this factorization into~\eqref{eq:dec_change_var} we obtain
\begin{align}
	&\Ex{\vectiid,\vinp}{\log\abs{\det{\frac{\partial \vecseg{\vectoutnn}{\setJ}}{\partial (\vectiid,\vecseg{\vinp}{\natseg{2}{\blocklength}})}}}}\nonumber\\
	&\quad=\sum_{j\in\natseg{\rankcov+2}{\blocklength}}\Ex{\vectiid}{\log\abs{\sum_{q=1}^{\rankcov} \miidc_{1q} \sqrtcovc_{j q}}}+\Ex{\vinp}{\log\abs{\det
\matM_1(\minp)}}\nonumber	\\
		&\qquad+\Ex{\vectiid}{\log\abs{\det\matM_2(\miid)}}+
	\log\abs{\det\matM_3(\sqrtcov)}\nonumber\\
		&\qquad+\Ex{\vectiid}{\log
	\abs{\det\matM_4(\miid)}}+
	\Ex{\vinp}{\log\abs{\det\matM_5(\minp)}}.
	\label{eq:expdecomp}
\end{align}
The first and the third term on the RHS of~\eqref{eq:expdecomp} are finite because $\miid$ has \iid Gaussian components. The fifth term is finite for the same reason, because $\log
\abs{\det\miid^{-1}}=-\log
\abs{\det\miid}$.
The second and the sixth term are finite because $\Ex{}{\log \abs{\vinpc_i}}>-\infty, i\in\natseg{1}{\blocklength}$, by assumption. 
Finally, we show in Appendix~\ref{app:M_three_full_rank} that  the matrix $\matM_3(\sqrtcov)$ has full rank if  \propspark is satisfied. 
This then implies that the fourth term on the RHS of~\eqref{eq:expdecomp} is also finite.

\section{Conclusion and Further Work}
% We showed that, surprisingly, the noncoherent-capacity pre-log of a temporally correlated SISO block-fading channel can be increased by adding antennas at the receiver side only. 
In this paper, we analyzed the noncoherent-capacity pre-log of a temporally correlated block-fading channel.
We showed that, surprisingly, the capacity pre-log in the SIMO case can be larger than that in the SISO case.
This result was established for the special case of the number of receive antennas being equal to the rank of the channel covariance matrix. Interesting open issues include extending the lower bound in Theorem~\ref{thm:mainLB} to an arbitrary number of receive antennas and finding a tight upper bound on the capacity pre-log.

\appendices
\section{}
\label{app:bijection}
\subsubsection*{Proof of Lemma~\ref{lem:bijection}}
We  need to show that the function $\vecseg{\vectoutnn}{\setJ}(\vectiid,\vecseg{\vinp}{\natseg{2}{\blocklength}})$ is  one-to-one a.e. Hence, we can exclude sets of measure zero from its domain. 
In particular, we shall consider the restriction of the function $\vecseg{\vectoutnn}{\setJ}(\vectiid,\vecseg{\vinp}{\natseg{2}{\blocklength}})$ to the set of pairs $(\vectiid,\vecseg{\vinp}{\natseg{2}{\blocklength}})$,  which satisfy 
\begin{inparaenum}[(i)]
\item
\label{it:i1}
$0 <\abs{\vinpc_i}<\infty$ for all $i\in\natseg{2}{\blocklength}$;
\item
\label{it:i2}
the matrix~$\miid$, defined in~\eqref{eq:miiddef} is invertible;
% \begin{equation}
% 	\label{eq:miiddef}
% 	\miid\define	
% 	\begin{bmatrix}
% 	\tp{\viid{1}}\\
% 	\vdots\\ 
% 	\tp{\viid{\RXante}}
% 	\end{bmatrix}
% \end{equation}
% \item
% \label{it:i3}
% the vector 
% \begin{equation}
% 	\label{eq:vectrdef}
% 	\vectr\define\tp{\bigl[
% 		\tp{\viid{1}}\tilde{\vecp} \ \underbrace{0  \ldots  0}_{\blocklength-1\ \text{times}} \ \tp{\viid{2}}\tilde{\vecp} \ \underbrace{0 \ldots 0}_{\rankcov-1\ \text{times}}\ \ldots\ \tp{\viid{\rankcov}}\tilde{\vecp} \ \underbrace{0 \ldots 0}_{\rankcov-1\ \text{times}}
% 	\bigr]}
% \end{equation}
% with $\tilde{\vecp}\define\tp{[\sqrtcovc_{11} \ldots \sqrtcovc_{1\rankcov}]}$ is nonzero;
\item
\label{it:i4}
the sum
$\sum_{q=1}^{\rankcov} \miidc_{1q} \sqrtcovc_{j q}$
is nonzero for all $j\in\natseg{\rankcov+2}{\blocklength}$.
\end{inparaenum}

To show that this restriction of the function $\vecseg{\vectoutnn}{\setJ}(\cdot)$ [which, with slight abuse of notation we still call $\vecseg{\vectoutnn}{\setJ}(\cdot)$] is one-to-one, we take an element $\tilde{\vecy}$ from its range and prove that the equation 
\begin{equation}
	\label{eq:firsteq}
	\vecseg{\vectoutnn}{\setJ}(\vectiid',\vecseg{\vinp'}{\natseg{2}{\blocklength}})=\tilde{\vecy}
\end{equation}
has a unique solution in the set of pairs $(\vectiid',\vecseg{\vinp'}{\natseg{2}{\blocklength}})$ satisfying the constraints~\eqref{it:i1}--\eqref{it:i4}. 
The element $\tilde{\vecy}$ can be represented as
$\tilde{\vecy}=\vecseg{\left(\left(\identity_{\RXante}\kron\minp\sqrtcov\right)\vectiid\right)}{\setJ}$, with $\minp=\diag\lefto(\tp{[\vinpc_1 \ \tp{\vecseg{\vinp}{\natseg{2}{\blocklength}}}]}\right)$,
where $(\vectiid,\vecseg{\vinp}{\natseg{2}{\blocklength}})$ satisfies the constraints~\eqref{it:i1}--\eqref{it:i4}.
Hence,~\eqref{eq:firsteq} can be rewritten in the following way
% 
% 
% Next, fix an element $\tilde{\vecy}$ from the image of the function $\vecseg{\vectoutnn}{\setJ}(\cdot)$. This element can be written as
% \begin{equation}
% 	\tilde{\vecy}=\vecseg{\left(\left(\identity_{\RXante}\kron\minp\sqrtcov\right)\vectiid\right)}{\setJ},\ \minp=\diag(\tp{[\vinpc_1,\tp{\vecseg{\vinp}{\natseg{2}{\blocklength}}}]}) 
% \end{equation}
% for some $(\vectiid,\vecseg{\vinp}{\natseg{2}{\blocklength}})\in \complexset^{\blocklength-1+\rankcov}$.
%   We want to show that $\tilde{\vecy}$ has a unique inverse under the transformation $\vecseg{\vectoutnn}{\setJ}(\cdot)$, i.e. the equation 
\begin{equation}
 \vecseg{\left(\left(\identity_{\RXante}\kron\minp'\sqrtcov\right)\vectiid'\right)}{\setJ}=\vecseg{\left(\left(\identity_{\RXante}\kron\minp\sqrtcov\right)\vectiid\right)}{\setJ} 
	\label{eq:bijection1}
\end{equation}
with $\minp'=\diag\lefto(\tp{[\vinpc_1 \  \tp{(\vecseg{\vinp'}{\natseg{2}{\blocklength}})}]}\right)$. 
To prove that~\eqref{eq:bijection1} has a unique solution, we follow the approach described in Section~\ref{sec:intuition} and convert~\eqref{eq:bijection1}  into a linear system of equations through a change of variables. 
In particular, thanks to constraint~\eqref{it:i1}, we can multiply both sides of~\eqref{eq:bijection1} by $\matseg{\identity_{\RXante}\kron\minp'}{\setJ}{\setJ}^{-1}$ and by $\matseg{\identity_{\RXante}\kron\minp}{\setJ}{\setJ}^{-1}$ and perform the substitution $\vinpinvc'_i=1/\vinpc'_i,\ i\in\natseg{2}{\blocklength}$. Then, we define $\vinpinv'=\tp{[\vinpinvc_2' \dots \vinpinvc_\blocklength']}$ and manipulate the equation such that all the unknowns are on one side of the equation and all the terms depending on the constant $\vinpc'_1$ are on the other side. 
These steps yield the following inhomogeneous linear system of equations
%  
% 
% Now~\eqref{eq:bijection1} is equivalent to 
% \begin{equation}
% 	\label{eq:bijection12}
%  \vecseg{\left(\left(\identity_{\RXante}\kron\minp^{-1}\sqrtcov\right)\vectiid'\right)}{\setJ}=\vecseg{\left(\left(\identity_{\RXante}\kron\minpinv'\sqrtcov\right)\vectiid\right)}{\setJ}
% \end{equation}
% where $\minpinv'=\diag(\vinpinv'),\ \vinpinv'=[\vinpinvc_1',\ldots,\vinpinvc_\blocklength']$. Using~\cite[Eq. (5), Sec. 7.2]{lutkepohl96}, it can be easily shown that~\eqref{eq:bijection12} holds if and only if
% \begin{equation}
% 	\label{eq:bijection2}
%  \begin{bmatrix}
%  	\matseg{\identity_{\RXante}\kron\minp^{-1}\sqrtcov}{\setJ}{\allel} &-\matseg{(\miid\kron \identity_\blocklength)\diagd{\sqrtcov}}{\setJ}{\allel}
%  \end{bmatrix}
% \begin{bmatrix}
% 	\vectiid'\\
% 	\vinpinv'
% \end{bmatrix}
% =\veczero
% \end{equation}
% with
% \begin{equation}
% 	\label{eq:Ddef}
% 		\diagd{\sqrtcov}=\begin{bmatrix}
% 		\diag([\sqrtcovc_{11}\dots \sqrtcovc_{\blocklength 1}])\\\vdots\\ \diag([\sqrtcovc_{1 \rankcov} \dots \sqrtcovc_{\blocklength \rankcov}]) 	
% 		\end{bmatrix}.
% 		% ,\,\,\text{and}\,\,
% 		% 	\miid=	
% 		% 	\begin{bmatrix}
% 		% 	\tp{\viid{1}}\\
% 		% 	\vdots\\ 
% 		% 	\tp{\viid{\RXante}}
% %\end{bmatrix}.
% \end{equation}
% Using that $\vinpinvc'_1=1/\vinpc'_1$ is fixed, equation \eqref{eq:bijection2} can be written explicitly in the homogeneous form
\begin{equation}
	\label{eq:bijection3}
 \matG
\begin{bmatrix}
	\vectiid'\\
	-\vinpinv'
\end{bmatrix}
=\frac{1}{\vinpc'_1}\vectr
\end{equation}
where 
\begin{align}
 	\label{eq:Gdef}
	\matG&= \begin{bmatrix}
	 	\matseg{\identity_{\RXante}\kron\minp^{-1}\sqrtcov}{\setJ}{\allel} &\matseg{(\miid\kron \identity_\blocklength)\diagd{\sqrtcov}}{\setJ}{\natseg{2}{\blocklength}}
	 \end{bmatrix}\\
%  \end{equation}
% and 
%  \begin{equation}
% 	\label{eq:vectrdef}
	\vectr&=\tp{\bigl[
		\tp\vecp_1\viid{1} \ \underbrace{0  \cdots  0}_{\blocklength-1\ \text{times}} \ \tp\vecp_1\viid{2} \ \underbrace{0 \cdots 0}_{\rankcov\ \text{times}}\ \cdots\ \tp\vecp_1\viid{\rankcov} \ \underbrace{0 \cdots 0}_{\rankcov\ \text{times}}
	\bigr]}\nonumber
\end{align}
and~$\tp\vecp_1$ denotes the first row of~$\sqrtcov$.
% and $\miid$ and $\vectr$ were defined in~\eqref{eq:miiddef} and~\eqref{eq:vectrdef}, respectively. 
%As~$\vectr/\vinpc'_1$ is a nonzero vector by assumption~\eqref{it:i3}, 
The solution of~\eqref{eq:bijection3} is unique if and only if $\det{\matG}\ne 0$. 
To establish this, it is useful to note that the matrix $\matG$ has a structure similar to the matrix on the RHS of~\eqref{eq:partialjacobian} [the only difference is that $\minp$ in~\eqref{eq:partialjacobian} 
is replaced by $\minp^{-1}$ in $\matG$].
Therefore, we can factorize $\det{\matG}$ in exactly the same way as the RHS of~\eqref{eq:partialjacobian}. 
Finally, we invoke the constraints~\eqref{it:i1},~\eqref{it:i2}, and~\eqref{it:i4} together with~\propspark to conclude that each term in the resulting factorization is nonzero. 
This completes the proof.
% 
% 
% 
% To see that this is, indeed, the case we note that the only difference between $\matG$ and the matrix in the RHS of~\eqref{eq:partialjacobian} is that $\minp$ in~\eqref{eq:partialjacobian} 
% is substituted by $\minp^{-1}$ in $\matG$. Therefore we can factorize $\det{\matG}$ just the way we factorized the RHS of~\eqref{eq:partialjacobianfact}. We can then envoke assumption~\eqref{it:i1}, \eqref{it:i2} and \eqref{it:i4} together with~\eqref{eq:sparkprop} to check that each term in the resulting product is nonzero, which completes the proof.

We point out that for $\blocklength=3$ and $\rankcov=2$, the matrix $\altmat$ defined in~\eqref{eq:inhomeq} is related to $\matG$ in \eqref{eq:Gdef} according to $\altmat=(\identity_2\kron\minp)\matG$, and therefore $\det{\altmat}\ne 0$ a.e.,  as claimed in Section~\ref{sec:intuition}.
%; the vector in the RHS of~\eqref{eq:inhomeq} is nonzero a.e., as  $\vectr/\vinpc'_1$ is here.  

%  invoke Lemma~\ref{lem:factrization}  to show that in Appendix~\ref{app:factorization}, $\det{\matG}$ can be factorized as follows 
% \begin{align}
% 	\label{eq:detfactor}
% 	&\abs{\det{\matG}}=\prod_{j\in\natseg{\rankcov+2}{\blocklength}}\abs{\sum_{q=1}^{\rankcov} \miidc_{1q} \sqrtcovc_{j q}}
% 		\abs{\det\matM_1(\minp^{-1})}
% 	\abs{\det\matM_2(\miid)}\nonumber\\
% 	&\qquad\qquad\times
% 		\abs{\det\matM_3(\sqrtcov)}
% 		\abs{\det\matM_4(\miid)}
% 		\abs{\det\matM_5(\minp^{-1})}
% \end{align}
% where $\matM_1$--$\matM_5$ are defined in~\eqref{eq:M1}--\eqref{eq:M5}, respectively.
% 
% 
% Using the assumption~\eqref{eq:sparkprop} on the matrix $\sqrtcovsub$, we show in Appendix~\ref{app:M_three_full_rank} that the matrix $\matM_3(\sqrtcov)$ is full rank and therefore the fourth term in the RHS of~\eqref{eq:detfactor} is nonzero. All the other terms in the RHS of~\eqref{eq:detfactor} are nonzero because we have restricted the domain of the function $\vecseg{\vectoutnn}{\setJ}(\cdot)$ as specified in~\ref{it:i1}--\ref{it:i4}.

% %%%%%%%%%%%%%%%%%%%%%%%%%%%%%%%%%%%%%
% %%%%%%%%%%%%%%%%%%%%%%%%%%%%%%%%%%%%%
%\section{Proof of~\eqref{eq:partialjacobianfact}}
\section{}
\label{app:factorization}
\subsubsection*{Proof of~\eqref{eq:partialjacobianfact}}
	As a consequence of the choice of~$\setJ$ in~\eqref{eq:setJ}, each of the last  $\blocklength-\rankcov-1$ columns of the matrix on the RHS of~\eqref{eq:partialjacobian} has exactly one nonzero element. This allows us to use the Laplace formula to expand the determinant along these columns iteratively to get
	\begin{align}
		% ~\eqref{eq:partialjacobian}
		%
	\label{eq:detdecomposition}
		&\abs{\det{\begin{bmatrix}
			\matseg{\identity_{\RXante}\kron\minp\sqrtcov}{\setJ}{\allel} & \matseg{(\miid\kron \identity_\blocklength)\diagd{\sqrtcov}}{\setJ}{\natseg{2}{\blocklength}}
		\end{bmatrix}}}\nonumber\\
		&\qquad\qquad=\prod_{j\in\natseg{\rankcov+2}{\blocklength}}\abs{\sum_{q=1}^{\rankcov} \miidc_{1q} \sqrtcovc_{j q}}
	\abs{\det(\matE)}
\end{align}
	where 
	\begin{align*}
	\matE&=	\begin{bmatrix}
		\identity_{\RXante}\kron\tilde{\minp}\sqrtcovsub & (\miid\kron \identity_{\rankcov+1}) 	\tilde{\matD}
		\end{bmatrix}\\
		\tilde{\matD}&= \matseg{\diagd{\sqrtcovsub}}{\allel}{\natseg{2}{\rankcov+1}},\qquad 
		\tilde{\minp}= \matseg{\minp}{\natseg{1}{\rankcov+1}}{\natseg{1}{\rankcov+1}}.
	\end{align*} 
	Next, using simple properties of the Kronecker product and exploiting the block-diagonal structure of $\tilde{\matD}$, we factorize $\matE$ into a product of simple terms:
	\begin{align}
		\matE&=\left(\identity_{\RXante}\kron\tilde{\minp}\right)(\miid\kron \identity_{\rankcov+1})\begin{bmatrix}
				\left(\identity_{\rankcov}\kron\tilde{\sqrtcov}\right) & \tilde{\matD}
		\end{bmatrix}\nonumber\\
		&\times\begin{bmatrix}
			(\miid^{-1}\kron \identity_\rankcov) & \matzero \\ 
				\matzero &  	\identity_{\rankcov}
		\end{bmatrix}
		\begin{bmatrix}
			\identity_{\rankcov^2} & \matzero \\ 
			\matzero &  \matseg{\tilde{\minp}}{\natseg{2}{\rankcov+1}}{\natseg{2}{\rankcov+1}}^{-1}
		\end{bmatrix}.
		\label{eq:factor1}
	\end{align}
		The proof is completed by inserting~\eqref{eq:factor1} into~\eqref{eq:detdecomposition}, and using the multiplicativity of the determinant.	
		
%\end{proof}	

	 \section{}%The Matrix $\matM_3$ in~\eqref{eq:M3} Is Full Rank}
	\label{app:M_three_full_rank}

	\begin{lem}
	\label{lem:spark}
		Let $\mat$ be an $(N+1)\times N$ matrix. If any set of $N$ rows of~$\mat$ is linearly independent, then the $N(N+1)\times N(N+1)$ matrix $\hat{\mat}$ defined as
		$\hat{\mat}=
		\begin{bmatrix}
				\left(\identity_{N}\kron\mat\right) &\! \matseg{\diagd{\mat}}{\allel}{\natseg{2}{N+1}} 
		\end{bmatrix}
		$
		has full rank.
	\end{lem}

\begin{proof}
	%	First note that  $\spark(\tp{\mat})=N+1$ implies that $\rank(\mat)=N$. 
	The proof is by contradiction. 
	Assume that $\hat{\mat}$ does not have full rank. Then there exists an $N(N+1)$-dimensional nonzero vector $\tp{\vectr}=\begin{bmatrix}\tp{\vectr_1}& \cdots & \tp{\vectr_N}\end{bmatrix}$, where $\vectr_n \in \complexset^{N+1}$, such that $\tp{\vectr} \hat{\mat}=\veczero$. 
Because 
$\hat{\mat}=\begin{bmatrix}
		\left(\identity_{N}\kron\mat\right) &\! \matseg{\diagd{\mat}}{\allel}{\natseg{2}{N+1}} 
\end{bmatrix}$, 
we have in particular that 
\begin{inparaenum}[(i)]
\item $\tp{\vectr}\left(\identity_{N}\kron\mat\right)=\veczero$ and  
\item $\tp{\vectr}\matseg{\diagd{\mat}}{\allel}{\natseg{2}{N+1}}=\veczero$. 
\end{inparaenum}
We next analyze these two equalities separately.
Equality~(i) can be restated as $\tp{\vectr_n}\mat=\veczero$ for all $n\in\natseg{1}{N}$, 
which implies that all vectors $\vectr_n$ lie in the kernel of the $N\times (N+1)$ matrix  $\tp{\mat}$. 
Because $\tp{\mat}$ has rank $N$, its kernel must be of dimension $1$. 
Hence,  all vectors $\vectr_n$ must be collinear, i.e., there exists a vector~$\altvectr$ and a set of $N$ constants $c_n$ such that $\vectr_n=c_n\altvectr$, for all $n\in\natseg{1}{N}$.
The vector~$\altvectr$ and at least one of the constants $c_n$ must be nonzero because $\vectr$ is nonzero. 
Furthermore, because $\tp{\altvectr}\mat=\veczero$, and because any set of~$N$ rows of $\mat$ is linearly independent by assumption, \emph{all} components of $\altvectr$ must be nonzero. 

We now use this property of~$\altvectr$ to analyze equality (ii), which can be restated as 
\begin{equation*}
	%\nonumber
	%\label{eq:daprop}
	\tp{\vectr}\matseg{\diagd{\mat}}{\allel}{\natseg{2}{N+1}}= [c_1\tp{\altvectr} \cdots \ c_N\tp{\altvectr}] \matseg{\diagd{\mat}}{\allel}{\natseg{2}{N+1}} =\veczero
\end{equation*}
or, after straightforward manipulations, as
\begin{equation*}
	%\nonumber
	\matseg{\diag(\altvectr)}{\natseg{2}{N+1}}{\natseg{2}{N+1}}\matseg{\mat}{\natseg{2}{N+1}}{\allel}\tp{[c_1 \cdots c_N]} =\veczero.
\end{equation*}
Because all the components of $\altvectr$ are nonzero, this last equality implies that
% The condition $\tp{\vectr} \hat{\mat}=\veczero$ also implies that
% \begin{equation}
% 	%\nonumber
% 	\label{eq:daprop}
% 	\tp{\vectr}\matseg{\diagd{\mat}}{\allel}{\natseg{2}{N+1}}= \begin{bmatrix}c_1\tp{\altvectr}& \ldots & c_N\tp{\altvectr}\end{bmatrix} \matseg{\diagd{\mat}}{\allel}{\natseg{2}{N+1}} =\veczero.
% \end{equation}
% 
% The assumption that $\spark(\tp{\mat})=N+1$ implies that each vector from the kernel of $\tp{\mat}$, and in particular the vector $\altvectr$, has {\it all} nonzero components. 
% %
% %
% %
% This allows us to remove $\altvectr$ from the equality in~\eqref{eq:daprop} and conclude, as a consequence of the block-diagonal structure of the matrix $\diagd{\mat}$, that
% %
%\begin{equation}
%	\nonumber
	$\matseg{\mat}{\natseg{2}{N+1}}{\allel}\tp{[c_1 \cdots c_N]} =\veczero.$
%\end{equation}
%
However, this contradicts the assumption that any set of $N$ rows of $\mat$ is linearly independent (recall that at least one of the constants $c_n$ is nonzero). 
Hence, $\hat{\mat}$ must have full rank.
\end{proof}

\bibliographystyle{IEEEtran}
\bibliography{IEEEabrv,publishers,confs-jrnls,vebib}

\end{document}